\documentclass[10pt]{article}
\usepackage{geometry}
\usepackage[ruled]{algorithm2e} 

\usepackage{amsmath,amssymb,amsfonts}
\usepackage{multicol}
\usepackage{cleveref}
\usepackage{graphicx}
\usepackage{textcomp}
\usepackage{xcolor}
\usepackage{theorem}
\usepackage{tabularx}
\usepackage{mathrsfs}
\usepackage{url}
\usepackage{subfigure}
\usepackage{tabularx}
\usepackage{float}
\usepackage{enumitem}
\usepackage{setspace}
\usepackage{booktabs}
\usepackage{threeparttable}
\usepackage{thmtools}
\usepackage{authblk}
\usepackage{caption}
\usepackage[misc]{ifsym}
\usepackage[numbers,sort&compress]{natbib}

\newtheorem{theorem}{Theorem}[section]
\newtheorem{assum}{Assumption}[section]

\newtheorem{lemma}{Lemma}[section]

\newtheorem{example}{Example}

\newcommand{\tx}{\textsf{tx}}
\newcommand{\set}[1]{\left\{#1\right\}}

\newtheorem{rema}{Remark}[section]

\newtheorem{define}{Definition}[section]

\newenvironment{proof}{{\indent\it Proof:\quad}}{\hfill $\square$\par}

\setcitestyle{authoryear}

\title{A Framework of Transaction Packaging in High-throughput Blockchains}
\author[$\textsuperscript{1}$]{Yuxuan Lu}
\author[$\textsuperscript{1}$]{Qian Qi}
\author[$\textsuperscript{2}$]{Xi Chen}

\affil[$\textsuperscript{1}$]{Center on Frontiers of Computing Studies, Peking University}
\affil[$\textsuperscript{2}$]{Leonard N. Stern School of Business, New York University}
\affil[ ]{\texttt{yx\_lu@pku.edu.cn,qiqian@pku.edu.cn,xc13@stern.nyu.edu}}

\date{}

\begin{document}
\maketitle
\begin{abstract}
We develop a model of coordination and allocation of decentralized multi-sided markets, in which our theoretical analysis is promisingly optimizing the decentralized transaction packaging process at high-throughput blockchains or Web 3.0 platforms. In contrast to the stylized centralized platform, the decentralized platform is powered by blockchain technology, which allows for secure and transparent Peer-to-Peer transactions among users. Traditional single-chain-based blockchains suffer from the well-known \emph{blockchain trilemma}. Beyond the single-chain-based scheme, decentralized high-throughput blockchains adopt parallel protocols to reconcile the blockchain trilemma, implementing any tasking and desired allocation. However, unneglectable network latency may induce partial observability, resulting in incoordination and misallocation issues for the decentralized transaction packaging process at the current high-throughput blockchain protocols.

To address this problem, we consider a strategic coordination mechanism for the decentralized transaction packaging process by using a game-theoretic approach. Under a tractable two-period model, we find a Bayesian Nash equilibrium of the miner's strategic transaction packaging under partial observability. Along with novel algorithms for computing equilibrium payoffs, we show that the decentralized platform can achieve an efficient and stable market outcome. The model also highlights that the proposed mechanism can endogenously offer a base fee per gas without any restructuration of the initial blockchain transaction fee mechanism. The theoretical results that underlie the algorithms also imply bounds on the computational complexity of equilibrium payoffs.
 
\end{abstract}
\clearpage

\section{Introduction}

\indent The goal of this paper is to develop a general theory of coordination and allocation in a decentralized multi-sided platform with heterogeneous users, which can be applied to the high-throughput blockchain platform or Web 3.0 platform design. To this end, we focus on the high-throughput blockchain in the present paper and propose a novel strategic transaction packaging mechanism based on the game-theoretic approach. The mechanism can also be applied to the incentive design, consensus mechanism design, and resource allocation in the decentralized multi-sided market.

\indent Since \cite{nakamoto2008bitcoin}, the rapid development of blockchain protocols put forth a new paradigm of the business ecosystem (says, Web 3.0 platform) and shed new light on the prospect of both computer science and platform economies. Most recently, high-throughput blockchain systems' rise in popularity has been inspired by the so-called \emph{blockchain trilemma}\footnote{The traditional single-chain-based blockchain protocols can only handle a limited number of transactions per second. For instance, the Bitcoin network can process around seven transactions per second. Suppose blockchain technology is to be adopted globally, primarily supporting the latest Web 3.0 platform. In that case, it should be able to handle enormous amounts of data and at faster transaction speeds so that more people can use the network without it becoming too slow or expensive to use.

\indent However, the fundamental design of many workhorse decentralized networks means that increasing scalability tends to weaken decentralization or security. This is what’s known as the blockchain trilemma. This term was popularized by \emph{Vitalik Buterin}, the co-founder of Ethereum. Thus, users must be aware of three elements desirable in a blockchain: decentralization, security, and scalability. Namely, the blockchain trilemma can also be recognized as the problem that it’s hard for single-chain blockchain protocols to achieve optimal levels of all three properties simultaneously. Increasing one usually leads to a weakening of another.} for the conventional single-chain-based blockchain protocols (e.g., Bitcoin (\cite{nakamoto2008bitcoin}) and Ethereum 1.0 (\cite{buterin2014next})). The main building blocks of the blockchain trilemma include

\begin{itemize}
    \item \textbf{Decentralization} offers the foundation of what is known as Web 3.0. In contrast to the Web 2.0 protocols that are full of networks controlled by centralized platforms, Web 3.0 is an internet where decentralized blockchain technology lets users control their data and online lives. 
    \item \textbf{Security} for the decentralized systems, especially the blockchain protocols, can also be recognized as solving the Byzantine Fault Tolerance problem (henceforth, BFT). Take Bitcoin as an example of decentralized blockchain security, the Bitcoin network adopts a combination of cryptography and a network consensus mechanism called Proof of Work (PoW), yet plays into the issue of scalability, as the PoW mechanism is secure but relatively slow due to the single-chain structure being preserved.
    \item \textbf{Scalability} refers to the goal of designing a blockchain that can support faster transaction speed or much more transactions per unit of time. Scale is notable if blockchain technology is to serve more comprehensive Web 3.0 platforms and possibly billions of users in the coming Web 3.0 era. This is where many high-throughput blockchain protocols still need to improve in retaining decentralization and security.
\end{itemize}

\indent Namely, the core spirit of high-throughput blockchain protocols aims to reconcile the widely existing blockchain trilemma problem that appears in the traditional single-chain-based protocols, thereby supporting the foundation of the Web 3.0 platform in a  more scalable and more practicable way. Specifically, the goal of the high-throughput blockchain systems is to provide a higher-performance distributed-consensus mechanism to improve on-chain transaction speed by replacing the single-chain structure yet hopefully retaining the decentralization and security features in the traditional single-chain-based protocols. Examples include Conflux (see \cite{li2020decentralized}), Algorand (see \cite{gilad2017algorand} and \cite{chen2019algorand}), and Cardano (see \cite{kiayias2017ouroboros} and \cite{kiayias2018ouroboros}), among many others. These high-throughput blockchain systems feature the block generation process (driven by network validators, henceforth, miners) in which the nodes rely on parallel communication to maximize their efficiency of block generation, which we also refer to as the parallel BFT consensus.\footnote{Many blockchain systems also adopt parallel BFT consensus in improving transaction speed, these include Solana (see \cite{yakovenko2018solana}), Fantom (see \cite{choi2018fantom}), and Bitgert Chain.} The challenge is to achieve consensus in such a high throughput scenario, ensuring that all nodes keep the same record (coordination) in their respective ledgers.

\indent Ideally, the existing representative high-throughput blockchains improve the on-chain transaction throughput via parallel BFT consensus\footnote{Including the DAG-based BFT consensus, and most recently see \cite{spiegelman2022bullshark}}. However, under the high network load scenario, high throughput can lead to severe coordination and resource allocation issues due to \emph{partial observability} driven by the parallel BFT-based high throughput consensus with network latency-derived user heterogeneity. These issues can result in limited usage of block capacity, thereby inducing a dramatic loss of network throughput. Now we describe the key mechanism of partial observability as follows:

\begin{itemize}
    \item \textbf{Unobservable Blocks} refers to a partial observability that appears in the high-throughput blockchain, especially when the block generation process is under the sub-second criterion. This is due to the fact that the block generation process is driven by the miners in the network, who are limited by the network latency and the block capacity upper bound. In this scenario, each miner cannot precisely observe which orders in the mempool have explicitly been put into a specific newly generated block, given that many parallel blocks can be generated simultaneously. This phenomenon can also be described within two periods and detailedly illustrated in \cref{fig:model}:
    \begin{itemize}
    \item Date $0$, traders submit orders to the mempool with a gas fee\footnote{A gas fee is a blockchain transaction fee paid by traders on the blockchain. This gas fee is automatically paid to miners for their services to the blockchain via smart contract. Without the gas fees, anyone would have no incentive to stake their thereby gain and help secure the network.}. 

    \item Date $1$, the miners generate blocks and pick some orders from the mempool, putting them into their blocks by maximizing their profit (gas fees). Given that miners cannot directly observe the orders in the block generated by other miners, some orders submitted to the mempool on Date $0$ may repeatedly be included in multiple blocks on Date $1$.
    \end{itemize}
    Namely, the miners can only observe the orders in the mempool on Date $0$, yet cannot observe which orders have been included in the blocks on Date $1$. This is the main source of partial observability. That is, in this case, the same order in the mempool can be repeatedly put into multiple blocks due to the effect of unobservable blocks.

    \item \textbf{Miscoordination} is the core economic consequence of the unobservable blocks effect. This miscoordination arises from miners' inability to observe which orders have been included in the blocks on Date $1$ due to the partial observability of the transaction packaging from other miners. As a result, miners have to make decisions on which orders to include in the blocks on Date $1$ without detecting which orders are already included in the blocks on Date $1$. This can lead to miners competing with each other to include the same orders in the blocks on Date $1$, resulting in repeated transaction packaging.
    
    \item \textbf{Misallocation} is another crucial economic consequence of the unobservable blocks effect. This misallocation issue arises from the miscoordination issues driven by partial observability. For instance, the miners in the platform need to allocate the resources (e.g., computing power) to generate blocks on Date $1$. However, due to the partial observability, miners cannot observe which orders have been included in the blocks on Date $1$ and hence cannot precisely allocate the block capacity to pursue the most profitable transactions. This can lead to the miners allocating too many resources to compete for the transactions with higher bidding, resulting in a dramatic loss of network throughput.

\end{itemize}

\begin{figure}[htbp]
\begin{centering}
\includegraphics[scale=.55]{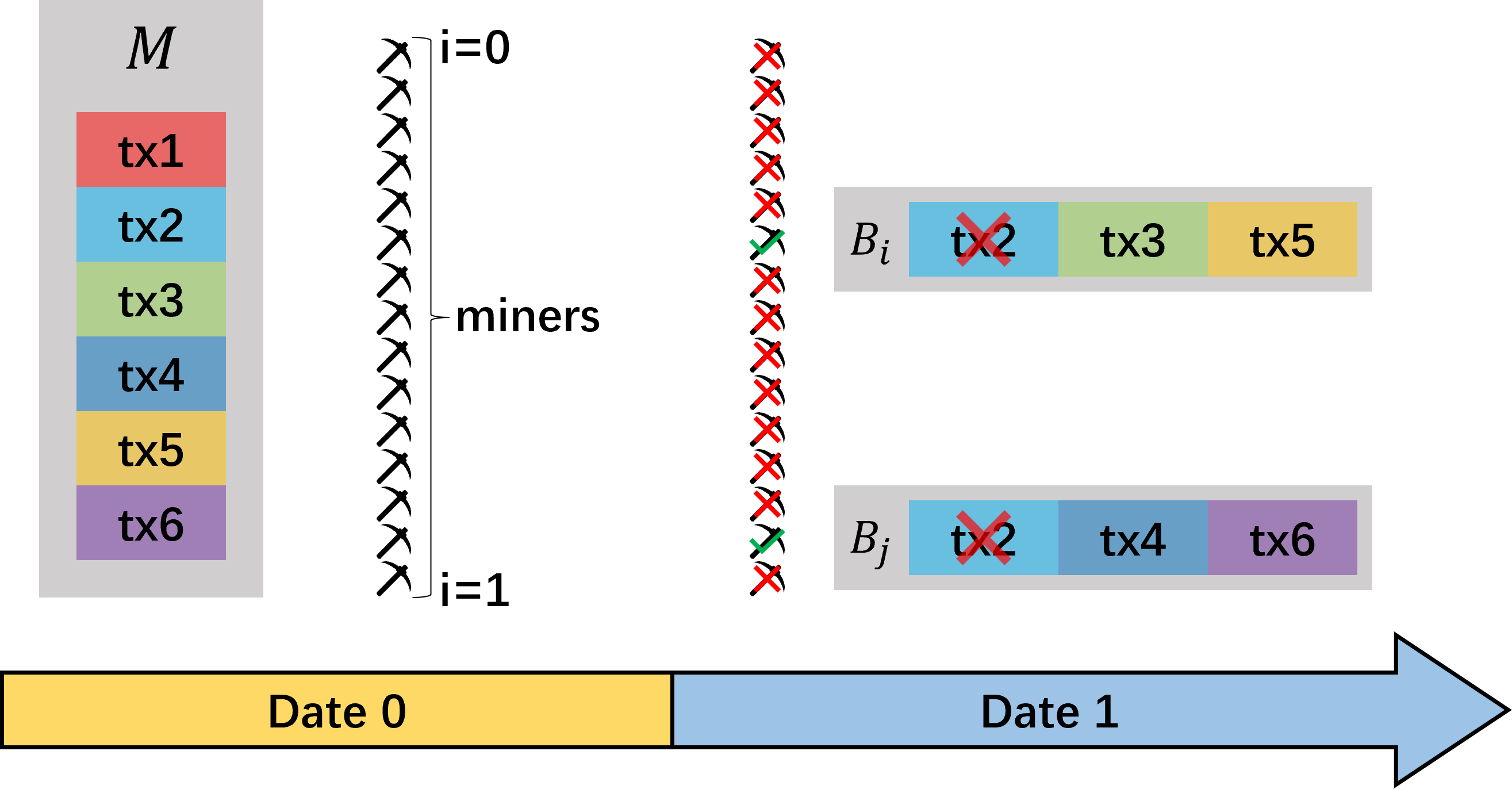}
\caption{Model description}
\label{fig:model}
\end{centering}
\end{figure}

\indent Again, the main purpose of the current paper is to develop a novel transaction packaging mechanism based on the game-theoretic approach to address the coordination and resource allocation issues in the high-throughput blockchain platform. Specifically, the proposed mechanism is designed to coordinate the users (miners and traders) in the platform, and to allocate the resources among the users in the system, to improve the system's throughput. The proposed mechanism is based on the game-theoretic approach, which allows the miners to package the transactions to maximize their profits strategically. 

\indent We emphasize the importance of the effect of ``unobservable blocks'' for the miners at high-throughput blockchains, especially when the block generation speed is under the sub-second level. Under partial observability, we provide the Nash equilibrium for the strategic transaction packaging game and calculate the equilibrium via novel algorithms. We also show that the proposed mechanism can improve the system throughput and reduce the transaction latency, compared to the existing transaction packaging strategies.

\indent Remarkably, we propose a mechanism to enable miners to coordinate their strategies without any external information acquisition. The proposed mechanism is based on a distributed ledger and relies on a reward system to incentivize miners to package transactions optimally. The reward system is designed to provide miners with a financial incentive to package transactions in an efficient manner. The reward system is designed to be coordinative and insightful, providing miners with the necessary strategies to make decisions about their transaction packaging under partial observability. Finally, the proposed mechanism allows the miners to collaborate in a decentralized manner and offers a way to scale the system to accommodate more transactions.

\indent We believe that our proposed model can benefit decentralized multi-sided markets in a number of ways. First, our model provides a theoretical framework for the efficient coordination and allocation of resources. By incentivizing users to form a cooperative policy, the model allows the multi-sided market to benefit from the blockchain's underlying features, including privacy, security, and transparency. Furthermore, our proposed model can increase transaction throughput, which is an important factor for the scalability of decentralized multi-sided markets. Finally, our model can be used to optimize the decentralized transaction packaging process, allowing for more efficient resource allocation and coordination.

\indent Overall, as the various blockchain protocols can engage in distinct service requirements as well as different characteristics and architectures, so far, a general framework for analyzing the transaction packaging mechanism of high throughput blockchain protocols has yet to be constructed particularly. Our proposed model provides a generalized solution to the coordination and allocation of decentralized multi-sided markets. We believe that our model can be used to improve the efficiency and scalability of various multi-sided platforms, and can be applied to a variety of blockchain-based applications.

\subsection{Literature Review}

\indent This paper sits at the confluence of several strands of literature, combining the multi-sided market with decentralized platforms and high throughput blockchains.

\indent Our paper relates to the literature on the platform economy and multi-sided markets. Note that large numbers of studies on the multi-sided market have been proposed, which starts with the seminal works by \cite{rochet2003platform}, \cite{rochet2006two}, \cite{caillaud2003chicken}, \cite{armstrong2006competition} for the two-sided market, and most recent contributions to this strand of literature include \cite{weyl2010price}, \cite{jullien2019information}, \cite{tan2021effects}, \cite{teh2022platform}, among many others. Most of the existing literature on multi-sided platforms has focused on both centralized monopolistic and duopolistic platforms. We focus on decentralized platform coordination and resource allocation issues in a model of heterogeneous users.

\indent Our study is related to recent literature on high throughput blockchain protocols. For instance, Conflux networks (see \cite{li2020decentralized}) adopt a Tree-Graph where each block references some other blocks, one of which being its parent block. Algorand (see \cite{gilad2017algorand}, \cite{chen2019algorand}) is another blockchain platform that is known for its high throughput. Specifically, the blockchain uses the Proof-of-Stake (PoS) mechanism, where validators are randomly picked and rewarded for their work. Cardano (see \cite{kiayias2017ouroboros} and \cite{kiayias2018ouroboros}) is a blockchain platform that focuses on community participation and security, in which the transactions are confirmed in just a few seconds. Solana (see \cite{yakovenko2018solana}) utilizes a proof-of-stake (PoS) consensus mechanism to gain popularity for offering high throughput transactions at low gas fees. Fantom (see \cite{choi2018fantom}) employs a directed acyclic graph (DAG) consensus mechanism to achieve high scalability. However, all these high-throughput blockchain protocols confront coordination and allocation issues as they no longer consider the strategic transaction packaging mechanism. We provide a theoretical analysis of this coordination and allocation for the transaction packaging game, in which the theory offers a set of equilibrium strategies for optimizing the coordination and efficiency of the existing high-throughput blockchain protocols.

\indent This paper also contributes to a large growing literature on Defi, FinTech, and blockchain in finance. These include \cite{cong2019blockchain}, \cite{biais2019blockchain}, \cite{easley2019mining}, \cite{schilling2019some}, \cite{chiu2019blockchain}, \cite{saleh2021blockchain}, \cite{cong2021tokenomics}, \cite{cong2021decentralized}, \cite{gryglewicz2021optimal}, \cite{prat2021equilibrium}, and \cite{cong2022token}, \cite{sockin2022decentralization}, among many others. Our analysis rests on a view of the decentralized platform's strategic transaction packaging game for coordination and allocation. The theoretical results will hopefully support the foundation for constructing the highly efficient Web 3.0 platform and the decentralized exchange (DEX).

\indent The remainder of this paper is organized as follows. \Cref{sec:model} introduces our model and assumptions. \Cref{sec:solution} characterizes the model solution of the two-period blockchain model. In \Cref{sec:discuss}, we describe the applications and extensions of our theory on an applicable blockchain protocol and highlight that the proposed mechanism can endogenously offer a base fee per gas without any restructuration of the blockchain transaction fee mechanism. \Cref{sec:conclu} concludes. The omitted proofs and algorithm details are presented in Appendix.

\section{Framework}\label{sec:model}

\indent We now formalize the consensus mechanism of our high-throughput blockchain protocol via a tractable two-period model. Specifically, we characterize the high-throughput blockchain in the sense of the conventional proof-of-work (PoW) and proof-of-stake (PoS) consensus mechanism. In particular, our framework can also be recognized as a multiparty asynchronous blockchain protocol.\footnote{Note that the multiparty asynchronous paradigm is one of the crucial patterns widely adopted in high-throughput blockchain systems so far, see \cite{li2020decentralized} for a representative case, among many others.} In a standard single-chain-based blockchain protocol, the specific time for the block generation process is crucial as the blockchain protocols must sort the generated blocks by timestamp to prevent double-spending problems. However, the standard scheme is no longer applicable to the high-throughput blockchain protocol as the blocks are multiparty asynchronous and generated simultaneously within the same timestamp. We formalize this specification of the existing high-throughput blockchain protocols, thereby analyzing several economic consequences in this section.

\subsection{The Model}

In conventional blockchains, miners may not try to replay the transactions in the previous blocks\footnote{Doing this is able to cause the block to fail the legality check. If a miner does this, she will be considered to launch an attack.}. However, miners cannot observe the whole history due to the network delay in high-throughput blockchains. Instead, they are unaware of blocks generated over a period of time (i.e., network latency).
As a result, it is natural for the same transaction to appear in different blocks. When a transaction appears in multiple blocks, the blockchain will execute the transaction in the earliest block and ignore its existence in other blocks.

The platform comprises two types of users: traders and miners. In contrast to the centralized platform\footnote{For example, the Web 2.0 platform came as a business revolution in the computer industry in the early 2000s. Platforms such as Amazon, Meta (Facebook), Twitter, YouTube, and others emphasize the flexibility of access, interaction, mobility, and participation. However, the centralization of Web 2.0 makes it easy for security issues, data gathering for malicious purposes, privacy intrusion, and cost as well. Otherwise, the centralized Web 2.0 platform is monopolistic, exacerbating wealth inequality and market manipulation. The central idea behind Web 3.0 is against this monopoly at the Web 2.0 platforms. With the large growing adoption of decentralized networks, blockchain is contributing as the backbone of Web 3.0, which led to a transition of the world wide web to a novel decentralized phase of use and service development.}, before any trader's submitted unpacked transaction $\tx\in M$ becomes a packaged legitimate transaction on the blockchain (or decentralized platform), it should stay in the mempool $M$, waiting for one of some potential decentralized miners to package it into a block. We now formalize these basic concepts in Definition \ref{def:2.1} to remove any ambiguity.

\begin{define}\label{def:2.1}
\indent 
\begin{itemize}

    \item\textbf{Trader} broadcasts transaction requests and hires a potential miner $i$ to work for packaging her orders as legitimate transactions $\tx$ into a specific block $B_{i}$, in which orders refer to the transactions in $\tau=0$.

    \item \textbf{Mempool} $M=\left\{1,2,...,m\right\}$ is the set of orders maintained by each miner. It acts as a waiting room for transactions that have not yet been included in a block. Note that after a transaction is packaged into a block, it should be removed from the mempool after a network latency.

    \item \textbf{Transaction} $\tx$ is an item in mempool $M$, representing a order submitted by a trader. When an order is broadcast, it is sent from a trader to all miners through the Peer-to-Peer network for packaging as a legitimate transaction. 

    \item \textbf{Gas price} is the bid per unit of block capacity occupied. We use the function $v(\cdot):M\rightarrow \mathbb{R}^+$ to represent the gas price bidding of transactions. For example, transaction $\tx$ has gas price $v(\tx)$. By defining the function $s(\cdot):M\rightarrow \mathbb{R}^+$ additionally as the capacity of the block occupied by transactions, the gas fee that transaction $\tx$ needs to pay to the miner is $s(\tx)\times v(\tx)$.

    \item \textbf{Miner} $i$ is responsible for generating blocks and packing the orders from mempool into her block $B_{i}$. Specifically, the miner is motivated chiefly by gas fees, so the orders with a higher gas price attached are those most likely to be packaged first. The blockchain platform comprises a continuum of miners with a measure of one unit, indexed by $i\in[0,1]$. 
  
    \item \textbf{Block} $B_{i}$ is a set of transactions to be packed by miner $i$. Each block has a capacity limit of $k$. Formally, $k$ is the sum of the maximum capacity used of all transactions that are accommodated in a block, that is, $\sum_{\tx\in B_i}s(\tx)\leq k$ for all $i\in [0,1]$.

\end{itemize}
\end{define}

To obtain a tractable model for theoretical analysis, we propose some assumptions of the blockchain protocols in practice. Specifically, we place additional constraints on network latency and block capacity usage of transactions.

\begin{assum}[Network latency]
Network latency is the time difference between when a block is broadcast and when miners receive it. In this paper, we simplify the effect of network latency and assume that all miners have the same network latency\footnote{This assumption can be relaxed to that the network latency of most miners is the same.}. We use a dimensionless parameter $\lambda$ to describe the network latency. The parameter $\lambda$ is defined as the number of blocks expected to be generated during the time interval of network latency. Additionally, we assume that miners cannot obtain information for any blocks generated within the range of their network latency.
\label{assum:network}
\end{assum}

\begin{assum}[Block capacity usage of transactions]
In the theoretical analysis, it is assumed that all transactions use fixed block capacity, i.e., $s(\tx)=1$ for each transaction $\tx$. As the total block capacity is $k$, each block can contain no more than $k$ transactions. In real-world blockchains, different transactions sometimes use various block capacities, which can also change. If these factors are taken into account in the model, at least the knapsack problem will be introduced. In \Cref{sec:realworld}, we will discuss how to extend our model solution to blockchains in practice.
\label{assum:capacity}
\end{assum}

The model lasts two dates $\tau\in\{0,1\}$. At date $0$, the orders arrive, and miners select the orders and hopefully package the selected orders into the coming block on their own. At date $1$, all blocks generate, and orders are becoming legitimate transactions via block packaging. The concrete timeline of the model is described in \Cref{fig:model}. 

\paragraph{Date $0$: orders arrive and orders selection}
At date $0$, miner $i$ maintains the public mempool $M$, which holds all submitted orders yet is still available for packaging. In the meanwhile, she selects some orders from mempool into her (not yet mined) block $B_{i}$ and works for mining her block $B_{i}$.

\paragraph{Date $1$: block generation}
At date $1$, a newly generated set $\Gamma$ comprises $\gamma$ random miners. The generation rule of the random miners follows a uniform distribution, in which each miner $i\in\Gamma$ can mine a block herself, and these $\gamma$ blocks can be appended to the blockchain. Note that $\gamma$, the number of miners that mine their blocks, is sampled by a Poisson distribution $\mathrm{Poisson}(\lambda)$. The reason that we choose $\mathrm{Poisson}(\lambda)$ is that a miner has no information about the blocks mined during her network latency in reality, and the distribution of the number of block mined in her network latency follows $\mathrm{Poisson}(\lambda)$ (see Assumption \ref{assum:network}).

\subsection{Strategy Space}\label{sec:strategy}
The miner's strategy space $\mathcal{S}$ is a subset of size no more than $k$ of all transactions in the current mempool, representing the transactions that are selected into the block. When $|M|\geq k$, we can ignore those strategies that contain less than $k$ transactions\footnote{Note that $|X|$ is the cardinality of set $X$.}, because they are clearly dominated. When $|M|<k$, the optimal strategy is obviously to package all transactions in $M$. Therefore, by taking $|M|\geq k$ as an underlying assumption, we only need to consider the strategy of packaging exactly $k$ transactions. Formally,
\begin{equation}
\mathcal{S}=\left\{D\big|D\subset M, |D|=k\right\}.
\end{equation}
For mix strategies, say $\Delta\mathcal{S}$ as the set of possible mixed strategies, i.e., the set of all probability distributions over $\mathcal{S}$. For each mixed strategy $\sigma$, we define $p^\sigma(D)$ as the probability that a mixed strategy $\sigma\in\Delta\mathcal{S}$ selects the set of transactions $B$. We additionally overload $p^\sigma(\tx)$ as the \emph{marginal probability} that transaction $\tx$ is selected. Formally, we have
\begin{equation}\label{eq:marginal}
p^\sigma(\tx)=\sum_{D\subset \mathcal{S}} p^\sigma(D)\times \mathbf{1}\left[\tx\in D\right],
\end{equation}
where $\mathbf{1}\left[\tx\in D\right]$ is a indicator function of that set $D$ includes transaction $\tx$, in which $\mathbf{1}\left[\tx\in D\right]=1$ represents $\tx\in D$, and otherwise $\tx\notin D$.

\subsection{Utility}

In reality, the utility of miners can be mainly divided into three parts: static block reward, gas fee, and miner extractable value (MEV).

\begin{itemize}
    \item \textbf{Static block reward} refers to the reward given by the blockchain consensus to the miners of the block. The value of static block reward may be related to factors such as the number of uncle blocks reported, but are not related to which transactions are included in the block. 
    \item \textbf{Gas fee} refers to the fee paid to the block creator by the transactions included in the block. Each transaction $\tx$ will specify a gas price $v(\tx)$. When it is packed into a block, the trader that submits the transaction will automatically pay a gas fee of $v(\tx)\times s(\tx)$, where $s(\tx)$ is the amount of block capacity usage of the transaction. In high-throughput blockchains, the same transaction may appear in more than one block. At this time, the transaction will be executed on the earliest block, and the transaction fee will be paid to the creator of this block.
    \item \textbf{Miner extractable value (MEV)} refers to the potential profit that can be extracted by miners through specific actions such as front-running, censorship, and reordering the transactions~\cite{daian2020flash}. Miner extractable value can be a significant source of revenue for miners~\cite{qin2022quantifying}.
\end{itemize}

\begin{assum}\label{assum:nomev}
We suppose that the miners can only obtain static block reward and gas fee to remove any impact of the miner extractable value (MEV) in the benchmark case, in which the case includes MEV can be added as a formal extension of our benchmark model.
\end{assum}

\begin{rema}
Note that we can ignore the static block reward as it is fixed and not correlated with the transactions packed. As a result, we define the utility of the miner $i$ who mined block $B_i$ as the expected gas fee received.   
\end{rema}

Recall that in high-throughput blockchains, when the same transaction appears in multiple blocks, it will be executed on the earliest block, and the transaction fee will be paid to the creator of this block. Therefore, a block $B_i$ can obtain the transaction fee of transaction $\tx$ if and only if none of the blocks in the network latency contains $\tx$, i.e., it needs to compete with $\mathrm{Poisson}(\lambda)$ blocks. In the model, conditioning on miner $i$ mines a block, the number of other blocks still follows $\mathrm{Poisson}(\lambda)$. Therefore, it is natural to define that block $B_i$ can receive transaction fees of transaction $\tx\in B_i$ if and only if $\tx$ does not appear in other blocks.

If a miner does not mine a block, its utility is $0$ surely. Formally, conditioning on $i\in \Gamma$, i.e. miner $i$ successfully mines a block, the utility $u_i$ is
\begin{equation}
u_i(\sigma_i,\sigma_{-i})=\mathrm{E}_{B_i\sim \sigma_i,B_{-i}\sim \sigma_{-i}}\left[\sum_{\tx \in B_i} v(\tx)\times \Pr\left[\nexists j\in\Gamma\text{ s.t. }\tx\in B_j\cap j\neq i\right]\right],
\end{equation}
where $\sigma_{-i}$ denotes the mixed strategies for all miners except the miner $i$. Notice that the term $\mathbb{\Pr}\left[\nexists j\in\Gamma\text{ s.t. }\tx\in B_j\cap j\neq i\right]$ measures the probability for a transaction $\tx$ has only been packaged in the unique block $B_{i}$. By applying the marginal probability we mentioned in \Cref{sec:strategy}, we can simplify the utility as the expected revenue that a miner gains conditioning she successfully mines her block:
\begin{equation}
u_i(\sigma_i,\sigma_{-i}) = \sum_{\tx\in M}p^{\sigma_i}(\tx)v(\tx)\times\Pr_{B_{-i}\sim \sigma_{-i}}\left[\nexists j\in\Gamma\text{ s.t. }\tx\in B_j\cap j\neq i\right].
\label{def:utility}
\end{equation}

\section{Model solution}\label{sec:solution}

In this section, we propose a Nash equilibrium of the game generated by the model in \Cref{sec:model}. The Nash equilibrium proposed can be embedded in the mining protocol of the blockchain. In order to encode the equilibrium as the mining protocol of the miners, it is required that the mixed strategies of all miners are the same. Otherwise, miners have to maintain a consensus on their labels. Fortunately, the Nash equilibrium has the same mixed strategy $\sigma^*$ for all miners.

\begin{define}[Equilibrium strategy] 
The set $\Delta \mathcal{S}$ is all the mixed strategies of each miner. Let $\sigma^*\in \Delta\mathcal{S}$ be a mixed strategy. The mixed strategy $\sigma^*$ is an equilibrium strategy if the strategy profile that every miner chooses $\sigma^*$ as her strategy meets a Nash equilibrium. Formally, the mixed strategy $\sigma^*$ is a equilibrium strategy if
\[
\forall i\in [0,1]\forall \sigma_i\in \Delta \mathcal{S}, u_i(\sigma_i,\sigma^*_{-i})\leq u_i(\sigma^*,\sigma^*_{-i}),
\] i.e., the utility maximization problem
\[
\forall i\in [0,1]\sigma^*\in\mathrm{arg}\max_{\sigma_i\in \Delta \mathcal{S}} u_i(\sigma_i,\sigma^*_{-i}),
\]
where $\sigma^*_{-i}$ is the strategy profile of all the miners except $i$ choosing mixed strategy $\sigma^*$.
\end{define}

\begin{theorem}[Main theorem]\label{thm:main}
For mempool $M$, there exists an equilibrium strategy $\sigma^*$, i.e., every miner applying $\sigma^*$ meets a Nash equilibrium. Moreover, the equilibrium strategy $\sigma$ can be obtained in time complexity $O(|M|\log |M|)$.
\end{theorem}
In \Cref{sec:proof1} and \Cref{sec:proof2}, we put forward a detailed proof for \Cref{thm:main}. Before starting the formal proof, we first take an overview of the proof sketch. First, we show that for any function $p(\cdot)$ satisfying some basic requirements, \Cref{alg:corresponding} can output a mixed strategy $\sigma$ with marginal probability $p^\sigma(\tx)=p(\tx)$ for all $\tx\in M$. Then, we give a specific function $\hat{p}(\cdot)$ whose corresponding strategy is an equilibrium strategy. In this step, the function $\hat{p}(\cdot)$ does not necessarily satisfy the basic requirements, i.e., it is possible that $\exists \tx\in M,\hat{p}(\tx)\notin [0,1]$. Finally, \Cref{thm:main} carefully modifies the calculated function to meet the basic requirements while maintaining the property of equilibrium. 

\subsection{Mapping Marginal Probability to Mixed Strategy}\label{sec:proof1}

\begin{lemma}[Corresponding strategy of marginal probability]
For any function $p(\cdot):M\rightarrow \mathbb{R}$ satisfying that (1) $\forall \tx\in M, 0\leq p(\tx)\leq 1$; and (2) $\sum_{\tx\in M} p(\tx) = k$, there exists a mixed strategy $\sigma\in\Delta\mathcal{S}$ whose marginal probability $p^\sigma(\tx)=p(\tx)$ for all $\tx\in M$. Moreover, strategy $\sigma$ can be obtained in the time complexity of $O(|M|)$. In the rest of the paper, we call strategy $\sigma$ as the \emph{corresponding strategy} of the function $p(\cdot)$.
\label{lemma:marginal}
\end{lemma}

\begin{algorithm}[h]
    \SetAlgoNoLine
    \KwIn{mempool $M$, marginal probability $p(\cdot)$, block capacity $k$, parameter $\lambda$.}
    \KwResult{mixed strategy $\sigma$.}
    $sum[0]\gets 0$; $dv[0]\gets 0$\;
    \For{$\tx\in 1,2,...,|M|$}
    {
        $sum[\tx]\gets sum[\tx-1]+p(\tx)$\;
        $sum[\tx]\gets sum[\tx]-\lfloor sum[\tx]\rfloor$\tcp*{$\lfloor x\rfloor$ is the largest integer not greater than $x$}
        $dv[\tx]\gets sum[\tx]$\;
    }
    Sort the elements in $dv$ from small to large; $S\gets \textbf{none}$\;
    \For{$j \in 1,2,...,|M|$}
    {
        $p\gets dv[j]-dv[j-1]$\;
        $A \gets \left\{\tx\left|sum[\tx-1] \leq dv[j] < sum[\tx]\right.\right\}$\;
        Add ``select transaction set $A$ with probability $p$'' into $\sigma$\;
    }
    \KwRet{$\sigma$}\;
    \caption{Mapping $p(\cdot)$ to its corresponding strategy}
    \label{alg:corresponding}
\end{algorithm}

\begin{proof}
\Cref{alg:corresponding} shows the method to obtain the mixed strategy $\sigma$. Specifically, the method has three steps. The first step is to treat each transaction as a line segment with a length of $p(\tx)$ and then arrange all the line segments in turn on the positive semi-axis. Next, generate a real number $r$ from $[0,1)$ uniformly at random. The last step is to output that the transactions selected are those line segments covering positions $\set{r+n|n=0,1,...,k-1}$ on the axis. 

There are two facts to prove, the first of which is that this method will select exactly $k$ transactions. This is because each line segment has a length no more than $1$, leading it to only cover at most one position in $\set{r+n|n=0,1,...,k-1}$. The second is that for each transaction $\tx$, the marginal probability $p^\sigma(\tx)=p(\tx)$. We denote the starting position of the line segment $\tx$ on the number axis as $L$. We select the transaction $\tx$ if and only if $L-\lfloor L\rfloor \leq r < L-\lfloor L\rfloor+1$ or $L-\lfloor L\rfloor-1 \leq r < L-\lfloor L\rfloor$. Hence, the marginal probability $p^\sigma(\tx)$, that the range of $r$ selecting $\tx$, is $p(\tx)$.
\end{proof}

\begin{example}\label{eg:corres}
\Cref{fig:corresponding} presents the procedure of \Cref{alg:corresponding} in detail. Specifically, in mempool $M$, there are $7$ transactions, and the block capacity $k=3$. In \Cref{table:corres}, the first column is the transaction id; The second column is the $p(\cdot)$ value of each transaction; The third column shows the covered interval of the segment created by each transaction; The fourth column is the $r$ value that the miner should package each transaction. When randomly selected $r\sim \mathrm{Uniform}(0,1)$ is $0.37$, the miner should select transactions $\set{2,3,5}$ to package.
\end{example}

\begin{figure}[htbp]
\centering
    \begin{minipage}{0.58\linewidth}
    \centering
    \includegraphics[scale=.43]{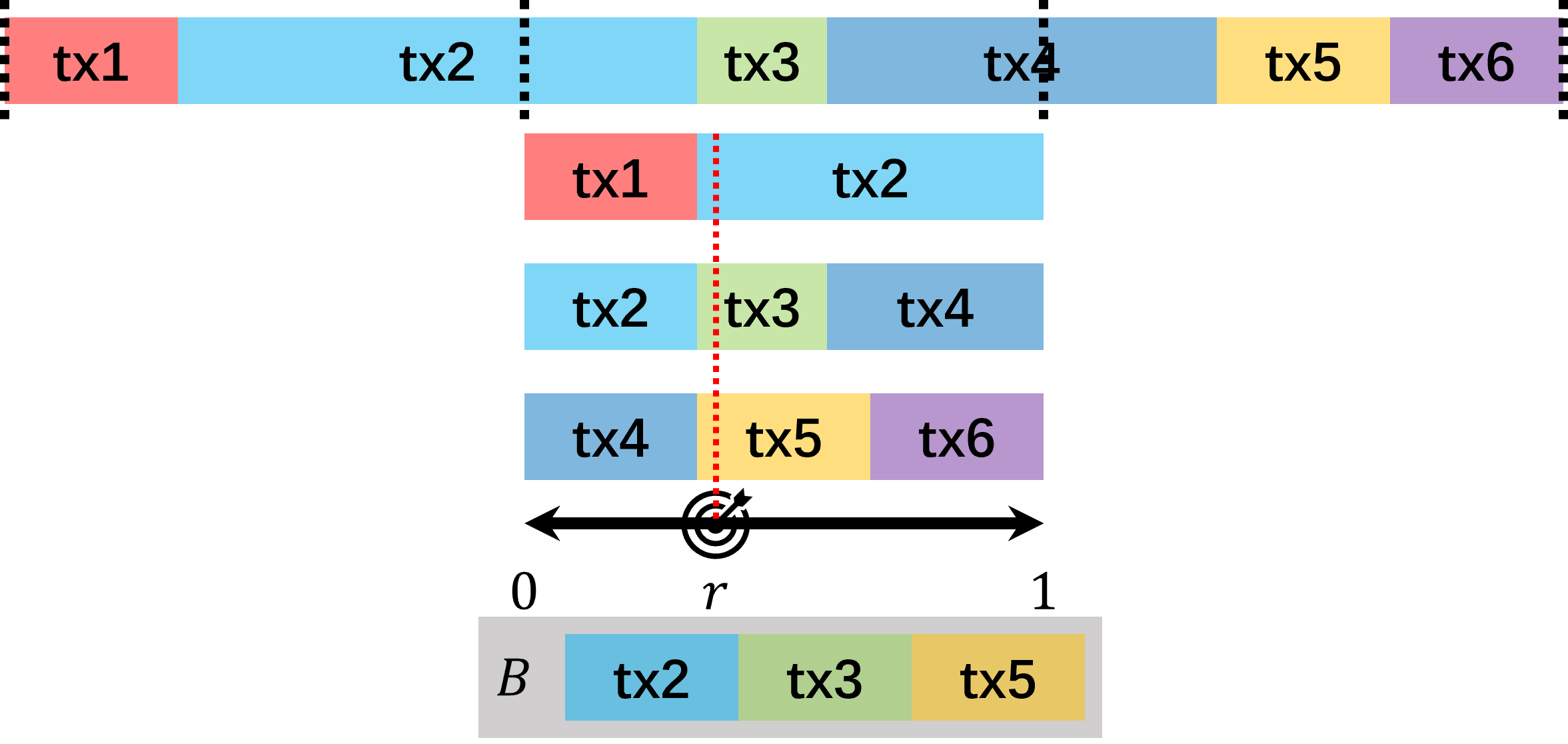}
    \caption{Finding the corresponding strategy of a function}
    \label{fig:corresponding}
    \end{minipage}
\hfill
    \begin{minipage}{0.41\linewidth}
    \centering\setlength{\tabcolsep}{4pt}
    \renewcommand\arraystretch{1.1}
    \begin{tabular}{|c|c|c|c|}
    \hline
    $\tx$ & $p(\tx)$ & interval & $r$ value\\
    \hline\hline
    1 & $ \frac{1}{3}$ & $ \left[0, \frac{1}{3}\right)$ & $ \left[0, \frac{1}{3}\right)$ \\\hline
    2 & $1$ & $ \left[ \frac{1}{3}, \frac{4}{3}\right)$ & $ \left[ \frac{1}{3}, 1\right) \cup \left[ 0,\frac{1}{3}\right) $ \\\hline
    3 & $ \frac{1}{4}$ & $ \left[ \frac{4}{3}, \frac{19}{12}\right)$ & $ \left[ \frac{1}{3}, \frac{7}{12}\right)$ \\\hline
    4 & $ \frac{3}{4}$ & $ \left[ \frac{19}{12}, \frac{7}{3}\right)$ & $ \left[ \frac{7}{12},1\right) \cup \left[0, \frac{1}{3}\right)$ \\\hline
    5 & $ \frac{1}{3}$ & $ \left[ \frac{7}{3}, \frac{8}{3}\right)$ & $ \left[ \frac{1}{3}, \frac{2}{3}\right)$ \\\hline
    6 & $ \frac{1}{3}$ & $ \left[ \frac{8}{3},3\right)$ & $ \left[ \frac{2}{3},1\right)$ \\\hline
    7 & $0$ & $\emptyset$ & $\emptyset$ \\\hline
    \end{tabular}
    \captionof{table}{Data table of \Cref{eg:corres}}
    \label{table:corres}
    \end{minipage}
\end{figure}

\subsection{Finding Equilibrium Strategy}\label{sec:proof2}

\begin{lemma}[Sufficient condition of equilibrium strategy]\label{lemma:eqcond}
A mixed strategy $\sigma^*$ is an equilibrium strategy if
\[
\exists w\in\mathbb{R} \forall \tx\in M,
\begin{cases}
\begin{aligned}
p^{\sigma^*}(\tx)=0 & \implies v(\tx)\exp\left(-\lambda p^{\sigma^*}(\tx)\right)\leq w\\
0<p^{\sigma^*}(\tx)<1 & \implies v(\tx)\exp\left(-\lambda p^{\sigma^*}(\tx)\right)=w\\
p^{\sigma^*}(\tx)=1 & \implies v(\tx)\exp\left(-\lambda p^{\sigma^*}(\tx)\right)\geq w
\end{aligned}
\end{cases},
\]
\end{lemma}
where $p^{\sigma}(\cdot)$ has been defined in \Cref{eq:marginal}. Note that $w$ is the threshold of the expected gas fee for each transaction.
\begin{proof}
Recall that $\Delta\mathcal{S}$ denotes the set of all possible mixed strategies for a miner, and $U(\sigma_i,\sigma_{-i})$ denotes the miner $i$'s utility for mining a block when her strategy is $\sigma_i$ and the rest miners' strategies are $\sigma_{-i}$.
We use $\tilde{v}(\tx)$ as shorthand for $v(\tx)\times \exp\left(-\lambda p^{\sigma^*}(\tx)\right)$ and suppose that strategy $\sigma^*$ satisfies the premise of the lemma. Afterwards, we can obtain that, for all $i\in [0,1]$, for all $\sigma\in\Delta\mathcal{S}$,
\begin{small}\begin{align*}
u_i(\sigma^*,\sigma^*_{-i}) = & \sum_{\tx\in M}p^{\sigma^*}(\tx)v(\tx)\times\Pr_{B_{-i}\sim \sigma^*}\left[\nexists j\in\Gamma\text{ s.t. }\tx\in B_j\cap j\neq i\right]\tag{\Cref{def:utility}}\\
= & \sum_{\tx\in M} p^{\sigma^*}(\tx)v(\tx)\times \mathrm{E}_{\gamma\sim \mathrm{Poisson}(\lambda)}\left(1-p^{\sigma^*}(\tx)\right)^\gamma\\
= & \sum_{\tx\in M} p^{\sigma^*}(\tx)v(\tx)\times \exp\left(-\lambda p^{\sigma^*}(\tx)\right)\\
= & \sum_{\tx\in M} p^{\sigma^*}(\tx)\tilde{v}(\tx)\\
= & \sum_{\tx\in M,0\leq p^{\sigma^*}(\tx)<1} p^{\sigma^*}(\tx)\tilde{v}(\tx) + \sum_{\tx\in M,p^{\sigma^*}(\tx)=1} \tilde{v}(\tx)\\
= & \sum_{\tx\in M,0\leq p^{\sigma^*}(\tx)<1} p^{\sigma^*}(\tx)w + \sum_{\tx\in M,p^{\sigma^*}(\tx)=1} \tilde{v}(\tx)\\
= & \sum_{\tx\in M,0\leq p^{\sigma^*}(\tx)\leq 1} p^{\sigma^*}(\tx)w + \sum_{\tx\in M,p^{\sigma^*}(\tx)=1} (\tilde{v}(\tx)-w)\\
= & kv + \sum_{\tx\in M,p^{\sigma^*}(\tx)=1} (\tilde{v}(\tx)-w)\\
\geq & \sum_{\tx\in M,0\leq p^{\sigma^*}(\tx)\leq 1} p^\sigma(\tx)w + \sum_{\tx\in M,p^{\sigma^*}(\tx)=1} p^\sigma(\tx)(\tilde{v}(\tx)-w)\\
= & \sum_{\tx\in M,0\leq p^{\sigma^*}(\tx)<1} p^\sigma(\tx)w + \sum_{\tx\in M,p^{\sigma^*}(\tx)=1} p^\sigma(\tx)\tilde{v}(\tx)\\
\geq & \sum_{\tx\in M,p^{\sigma^*}(\tx)=0} p^\sigma(\tx)\tilde{v}(\tx) + \sum_{\tx\in M,0<p^{\sigma^*}(\tx)<1} p^\sigma(\tx)\tilde{v}(\tx) + \sum_{\tx\in M,p^{\sigma^*}(\tx)=1} p^\sigma(\tx)\tilde{v}(\tx)\\
= & \sum_{\tx\in M} p^\sigma(\tx)\tilde{v}(\tx)\\
= & \sum_{\tx\in M}p^\sigma(\tx)v(\tx)\times\Pr_{B_{-i}\sim \sigma^*}\left[\nexists j\in\Gamma\text{ s.t. }\tx\in B_j\cap j\neq i\right]\\
= & u_i(\sigma,\sigma^*_{-i}).
\end{align*}\end{small}
By definition of Nash equilibrium, a strategy $\sigma^*$ is an equilibrium strategy if and only if $\forall i\in [0,1]\forall \sigma\in\Delta\mathcal{S}, u_i(\sigma^*,\sigma^*)\geq u_i(\sigma,\sigma^*)$. So, strategy $\sigma^*$ is an equilibrium strategy.

\end{proof}

\begin{lemma}[Weak equilibrium strategy]\label{lemma:weak}
For a mempool $M$, define function $\hat{p}(\cdot):M\rightarrow\mathbb{R}$ as
\begin{equation}\label{eq:phat}
\hat{p}(\tx)= \frac{1}{|M|}\left(k+\sum_{\tx'\in M} \frac{\ln v(\tx)-\ln v(\tx')}{\lambda}\right),
\end{equation}
where $\tx'$ enumerates every order in the mempool $M$. If $\hat{p}(\cdot)$ satisfies that $\forall \tx\in M, 0\leq \hat{p}(\tx)\leq 1$, then the corresponding strategy of function $\hat{p}(\cdot)$ is an equilibrium strategy. Moreover, 
\[
\forall\tx\in M, v(\tx)\exp\left(-\lambda \hat{p}(\tx)\right)=\exp\left(-\frac{\lambda k}{|M|}\right)\exp\left(\frac{1}{|M|}\sum_{\tx'\in M}\ln v(\tx')\right).
\]
\end{lemma}

\begin{proof}
First, we need to verify that $\hat{p}(\cdot)$ corresponds to a mixed strategy. By \Cref{lemma:marginal}, $\hat{p}(\cdot)$ should satisfy that $\forall \tx\in M, 0\leq \hat{p}(\tx)\leq 1$ and $\sum_{\tx\in M} \hat{p}(\tx) = k$. The first requirement exists in the suppose, so we only need to check the second.
\begin{small}\begin{align*}
\sum_{\tx\in M} \hat{p}(\tx) = & \sum_{\tx\in M} \frac{1}{|M|}\left(k+\sum_{\tx'\in M} \frac{\ln v(\tx)-\ln v(\tx')}{\lambda}\right)\\
= & \left(\sum_{\tx\in M}\frac{1}{|M|}k\right)+\frac{1}{|M|}\sum_{\tx\in M} \sum_{\tx'\in M} \frac{\ln v(\tx)-\ln v(\tx')}{\lambda}\\
= & k + \frac{1}{|M|}\sum_{\tx\in M} \left(\frac{|M|\ln v(\tx)}{\lambda}\sum_{\tx'\in M} \frac{-\ln v(\tx')}{\lambda}\right)\\
= & k + \frac{1}{|M|}\left(\left(\sum_{\tx\in M} \frac{|M|\ln v(\tx)}{\lambda}\right)-\left(|M|\sum_{\tx'\in M} \frac{-\ln v(\tx')}{\lambda}\right)\right)\\
= & k.
\end{align*}\end{small}

Now, we can briefly verify whether $\hat{p}(\cdot)$ fulfills the requirement of \Cref{lemma:eqcond} to complete the proof. In fact, we will show that $v(\tx)\exp\left(-\lambda \hat{p}(\tx)\right)$ is exactly the same for each $\tx\in M$:
\begin{small}\begin{align*}
v(\tx)\exp\left(-\lambda \hat{p}(\tx)\right) = & v(\tx)\exp\left(- \frac{\lambda}{|M|}\left(k+\sum_{\tx'\in M} \frac{\ln v(\tx)-\ln v(\tx')}{\lambda}\right)\right)\\
= & v(\tx)\exp\left(-\frac{\lambda k}{|M|}\right)\exp\left(-\frac{\lambda}{|M|}\sum_{\tx'\in M} \frac{\ln v(\tx)-\ln v(\tx')}{\lambda}\right)\\
= & v(\tx)\exp\left(-\frac{\lambda k}{|M|}\right)\exp\left(-\frac{1}{|M|}\left(|M|\ln v(\tx)-\sum_{\tx'\in M}\ln v(\tx')\right)\right)\\
= & v(\tx)\exp\left(-\frac{\lambda k}{|M|}\right)\exp\left(-\ln v(\tx)+\frac{1}{|M|}\sum_{\tx'\in M}\ln v(\tx')\right)\\
= & v(\tx)\exp\left(-\frac{\lambda k}{|M|}\right)\frac{1}{v(\tx)}\exp\left(\frac{1}{|M|}\sum_{\tx'\in M}\ln v(\tx')\right)\\
= & \exp\left(-\frac{\lambda k}{|M|}\right)\exp\left(\frac{1}{|M|}\sum_{\tx'\in M}\ln v(\tx')\right).
\end{align*}\end{small}
It can be observed that the right-hand side of the formula is a constant related to $M$ but not to $\tx$. By \Cref{lemma:eqcond}, the corresponding strategy of function $\hat{p}(\cdot)$ is an equilibrium strategy.
\end{proof}

Given that \Cref{lemma:weak} implies that the function $\hat{p}(\cdot)$ can be a marginal probability $p^{\sigma^*}(\cdot)$ of the equilibrium strategy $\sigma$ if and only if the $\forall\tx\in M, \hat{p}(\tx)\in [0,1]$. Namely, $\hat{p}(\cdot)$ can generate unreasonable results beyond the interval $[0,1]$, that is, we cannot always obtain an equilibrium strategy for every given mempool $M$. To fix this problem, we push the frontier of the \Cref{lemma:weak} by proposing the following theorem.

\begin{theorem}[Equilibrium strategy]\label{th:algeq}
Define a class of functions $p_x(\cdot):M\rightarrow\mathbb{R}$ with the subscript $x\in\mathbb{R}$ as
\[
p_x(\tx)=\begin{cases}
0 & \hat{p}(\tx)\leq x\\
\hat{p}(\tx)-x & x<\hat{p}(\tx)<x+1\\
1 & \hat{p}(\tx)\geq x+1
\end{cases},
\] where function $\hat{p}(\cdot)$ is defined in \Cref{lemma:weak}. Let $\hat{x}\in\mathbb{R}$ to be the smallest\footnote{The solution of the equation may not be unique, yet any solution $\hat{x}$ will lead to the same function $p_x(\cdot)$.} solution of equation $\sum_{\tx\in M}\min(\max(\hat{p}(\tx)-x,0),1)=k$. Then, function $p_{\hat{x}}(\cdot)$ satisfies that (1) $\forall\tx\in M, p_{\hat{x}}(\tx)\in [0,1]$; and (2) $\sum_{\tx\in M}p_{\hat{x}}(\tx)=k$. The corresponding strategy $\sigma$ of $p_{\hat{x}}(\cdot)$ is an equilibrium strategy\footnote{We can calculate the solution $\hat{x}$ efficiently. See detail in \Cref{sec:algcalcx}.}.
\end{theorem}

\begin{example}\label{eg:findx}
\Cref{fig:findxhat} presents the procedure of the calculating the marginal probability $p_{\hat{x}}(\cdot)$ in detail. Specifically, in mempool $M$, there are $7$ transactions, and the block capacity $k=3$. The first column of \Cref{table:corres} is the transaction id. The second column is each transaction's gas price $v(\cdot)$. The third column is the $\hat{p}(\cdot)$ value calculated according to \Cref{lemma:weak} and the fourth column is the marginal probability $p_{\hat{x}}(\cdot)$ after we solve the equation $\sum_{\tx\in M}\min(\max(\hat{p}(\tx)-x,0),1)=k$ to obtain $\hat{x}=\frac{1}{3}$.
\end{example}

\begin{figure}[htbp]
\centering
    \begin{minipage}{0.63\linewidth}
    \centering
    \includegraphics[scale=.43]{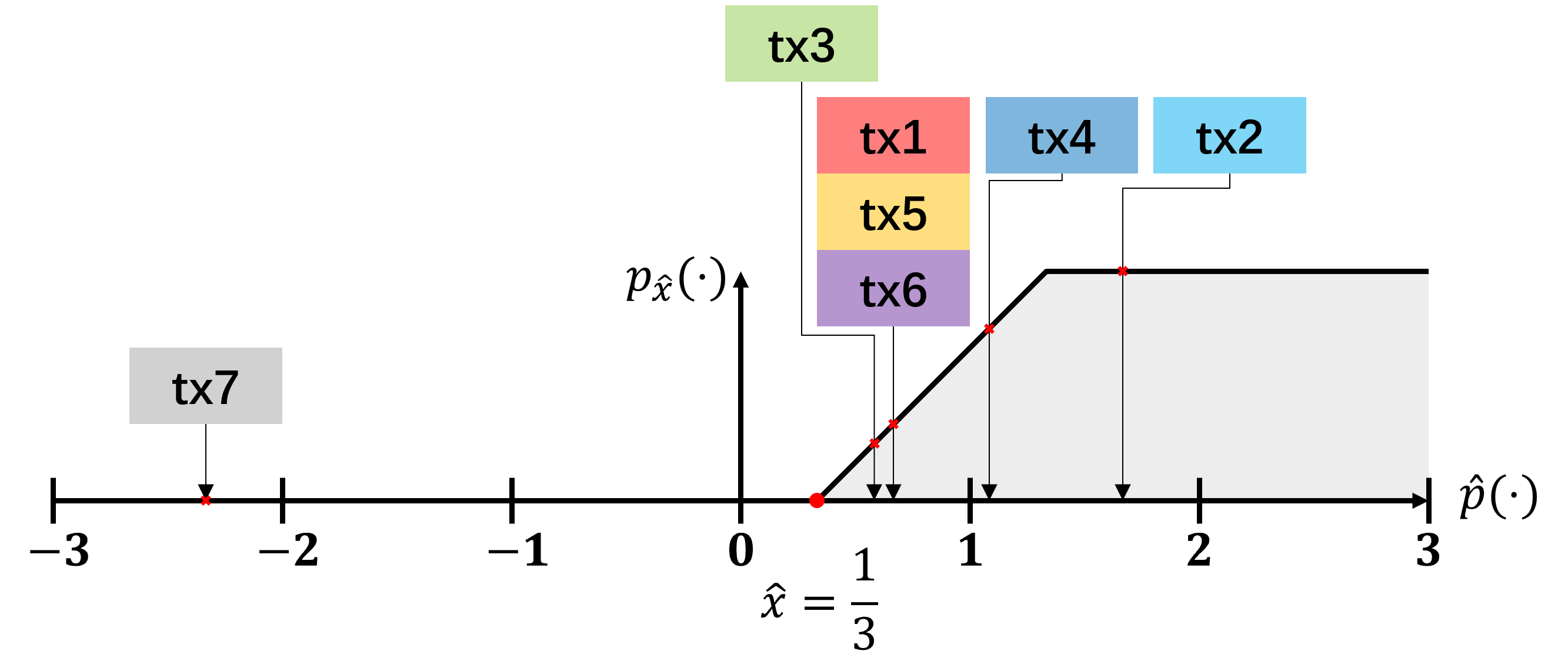}
    \caption{Calculating the marginal probability $p_{\hat{x}}(\cdot)$}
    \label{fig:findxhat}
    \end{minipage}
\hfill
    \begin{minipage}{0.36\linewidth}
    \centering\setlength{\tabcolsep}{4pt}
    \renewcommand\arraystretch{1.1}
    \begin{tabular}{|c|c|c|c|}
        \hline
        $\tx$ & $v(\tx)$ & $\hat{p}(\tx)$ & $p_{\hat{x}}(\tx)$\\\hline\hline
        1 & $1$ & $ \frac{2}{3}$ & $ \frac{1}{3}$\\\hline
        2 & $\exp\left(1\right)$ & $\frac{5}{3}$ & $1$\\\hline
        3 & $\exp\left(-\frac{1}{12}\right)$ & $ \frac{7}{12}$ & $\frac{1}{4}$\\\hline
        4 & $\exp\left(\frac{5}{12}\right)$ & $ \frac{13}{12}$ & $\frac{3}{4}$\\\hline
        5 & $1$ & $ \frac{2}{3}$ & $ \frac{1}{3}$\\\hline
        6 & $1$ & $ \frac{2}{3}$ & $ \frac{1}{3}$ \\\hline
        7 & $\exp(-3)$ & $ -\frac{7}{3}$ & $ 0$\\\hline
    \end{tabular}
    \captionof{table}{Data table of \Cref{eg:findx}}
    \label{table:weak}
    \end{minipage}
\end{figure}

\begin{proof}
The proof contains three parts. In the first part of the proof, we will show that the equation $\sum_{\tx\in M}\min(\max(\hat{p}(\tx)-x,0),1)=k$ always has a solution and the method to calculate it. Because the term $\min(\max(\hat{p}(\tx)-x,0),1)$ is continuous and non-increasing with $x$ for all $\tx\in M$, we have that the term $\sum_{\tx\in M}\min(\max(\hat{p}(\tx)-x,0),1)$ is continuous and non-increasing with $x$. According to the facts that 
\[
\begin{cases}
\lim_{x\rightarrow -\infty} \sum_{\tx\in M}\min(\max(\hat{p}(\tx)-x,0),1)=|M|\geq k\\
\lim_{x\rightarrow +\infty} \sum_{\tx\in M}\min(\max(\hat{p}(\tx)-x,0),1)=0
\end{cases},
\]
we can assert that the equation $\sum_{\tx\in M}\min(\max(\hat{p}(\tx)-x,0),1)=k$ has a solution. To calculate the solution $\hat{x}$, we can use a standard binary search algorithm to obtain the time complexity of $O(|M|\log |M|)$. The pseudo-code of the binary search algorithm is shown in \Cref{sec:algcalcx}.

The second part is to show that function $p_{\hat{x}}(\cdot)$ satisfies that (1) $\forall\tx\in M, p_{\hat{x}}(\tx)\in [0,1]$; and (2) $\sum_{\tx\in M}p_{\hat{x}}(\tx)=k$. Both of these facts are clear. Since the values of the three cases in the definition of the function $p_x(\cdot)$ are all in the interval $[0,1]$, the first requirement is satisfied. Because $\hat{x}$ is a solution of equation $\sum_{\tx\in M}\min(\max(\hat{p}(\tx)-x,0),1)=k$, the second requirement is satisfied. Thus, function $p_{\hat{x}}(\cdot)$ is the marginal probability of a mixed strategy $\sigma$.

Finally, we will prove that the corresponding strategy $\sigma$ of the function $p_{\hat{x}}(\cdot)$ satisfies the condition of \Cref{lemma:eqcond}. In terms of \Cref{lemma:eqcond} and \Cref{lemma:weak}, we set
\[
w = \exp\left(-\frac{\lambda k}{|M|}\right)\exp\left(\frac{1}{|M|}\sum_{\tx'\in M}\ln v(\tx')\right)\exp(-\lambda \hat{x})
\] and obtain that for all $\tx\in M$,
\begin{align*}
& v(\tx)\exp\left(-\lambda \left(\hat{p}(\tx)-x\right)\right)\\
= & v(\tx)\exp\left(-\lambda \hat{p}(\tx)\right)\exp(-\lambda x)\\
= & \exp\left(-\frac{\lambda k}{|M|}\right)\exp\left(\frac{1}{|M|}\sum_{\tx'\in M}\ln v(\tx')\right)\exp(-\lambda \hat{x})\tag{\Cref{lemma:weak}}\\
= & w.
\end{align*}
Therefore, we derive
\[
\begin{cases}
p_{\hat{x}}(\tx)=0 \implies v(\tx)\exp\left(-\lambda p_{\hat{x}}(\tx)\right)<v(\tx)\exp\left(-\lambda \left(\hat{p}(\tx)-\hat{x}\right)\right)=w\\
0<p_{\hat{x}}(\tx)<1 \implies v(\tx)\exp\left(-\lambda p_{\hat{x}}(\tx)\right)=v(\tx)\exp\left(-\lambda \left(\hat{p}(\tx)-\hat{x}\right)\right)=w\\
p_{\hat{x}}(\tx)=1 \implies v(\tx)\exp\left(-\lambda p_{\hat{x}}(\tx)\right)>v(\tx)\exp\left(-\lambda \left(\hat{p}(\tx)-\hat{x}\right)\right)=w,
\end{cases}
\] thereby meeting the condition of \Cref{lemma:eqcond}. So, the corresponding strategy $\sigma$ of function $p_{\hat{x}}(\cdot)$ is an equilibrium strategy.

\end{proof}

We can directly obtain the main theorem according to \Cref{lemma:marginal} and \Cref{th:algeq}.
\section{Discussions}\label{sec:discuss}
\subsection{Applications and Extensions}\label{sec:realworld}
In \Cref{sec:solution}, we discuss the equilibrium strategy in our model of a high-throughput blockchain. In this section, we focus on the extensions for \Cref{lemma:marginal} and \Cref{thm:main} to obtain generalized transaction packaging strategies for blockchains in practice. In our model, we assume that $s(\tx)=1$ for all transitions. By contrast, different transactions use various block capacities in practice. For instance, a block may contain thousands of token transfers, yet it may be occupied by a small number of transactions with nested calls of specifications in smart contracts. 

Recall that our model assumes every transaction has the same block capacity usage. Without this assumption, \Cref{lemma:marginal} can no longer hold. Therefore, we can empirically adopt the following randomization method instead of \Cref{alg:corresponding} in \Cref{lemma:marginal}: 

First, take a block capacity $k'$ slightly below the limit $k$, and then calculate a marginal probability $p(\cdot)$ of transaction selection, such that $\sum_{\tx\in M} s(\tx)p(\tx)=k'$. Finally, continue to randomly select each transaction $\tx$ according to the probability $p(\tx)$ independently until the selected transaction set $D$ satisfies $k-2k'\leq\sum_{\tx\in D} s(\tx)\leq k$.

After solving the difficulty of converting the marginal probability of transaction selection into a packing strategy, we only need to introduce $s(\tx)$ into \Cref{th:algeq}. Here, we propose the following amended version of \Cref{th:algeq} and attach the algorithm and proof in \Cref{sec:proofreal}.

\begin{restatable}{theorem}{thmainreal}\label{th:mainreal}
For a mempool $M$, define function $p_x(\cdot):M\rightarrow\mathbb{R}$ with the subscript $x$ as
\[
p_x^{\mathrm{real}}(\tx)=\begin{cases}
0 & \hat{p}^{\mathrm{real}}(\tx)\leq x\\
\hat{p}^{\mathrm{real}}(\tx)-x & x<\hat{p}^{\mathrm{real}}(\tx)<x+1\\
1 & \hat{p}^{\mathrm{real}}(\tx)\geq x+1
\end{cases},
\]where\footnote{Enumerators $\tx''$ and $\tx''$ enumerates every order in the mempool $M$.}
\begin{equation}\label{eq:phatreal}
\hat{p}^{\mathrm{real}}(\tx)= \frac{1}{\sum_{\tx'\in M} s(\tx')}\left(k+\sum_{\tx''\in M} s(\tx'')\frac{\ln v(\tx)-\ln v(\tx'')}{\lambda}\right).
\end{equation}
Let $\hat{x}$ to be the smallest solution of equation $\sum_{\tx\in M}\min(\max(\hat{p}^{\mathrm{real}}(\tx)-x,0),1)=k$. Then, function $p_{\hat{x}}(\cdot)$ satisfies that (1) $\forall\tx\in M, p_{\hat{x}}^{\mathrm{real}}(\tx)\in [0,1]$; and (2) $\sum_{\tx\in M}p_{\hat{x}}^{\mathrm{real}}(\tx)s(\tx)=k$. If there is a mixed strategy $\sigma$ such that $\forall \tx\in M,p^\sigma(\tx)=p(\tx)$, then $\sigma$ is an equilibrium strategy\footnote{In practice, we use $k'$ instead of $k$ and obtain an approximate equilibrium.}.
\end{restatable}

\subsection{Base Fee}
On August 5th, 2021, Ethereum activated a major upgrade called the London hard fork to implement EIP-1559 (and several other EIPs). After EIP-1559 was implemented, the blockchain consensus maintained an additional parameter called base fee. Base fee indicates the minimum gas price an order need to bid in each
block. According to the description in the EIP-1559 proposal, the existence of the base fee can help users better estimate transaction fees, thereby reducing or eliminating the ``mismatch between volatility of transaction fee levels and social cost of transactions''.

EIP-1559 made a qualitative change to Ethereum's transaction fee mechanism in order to provider a base fee \cite{leonardos2021dynamical,roughgarden2021transaction,10.1145/3548606.3559341}. Surprisingly, the transaction packaging strategy provided in this paper can lead to an natural base fee parameter, without any modification to the transaction fee mechanism.

Imagine a virtual transaction $\tx$ with block capacity usage $s(\tx)=0$ and add it to the mempool $M$. Since $s(\tx)=0$, it will not affect $\hat{x}$ and $\hat{p}^{\mathrm{real}}(\cdot)$ values of the remaining transactions (see \Cref{eq:phatreal}). Then, when this virtual transaction $\tx$ satisfies the equation $\hat{p}^{\mathrm{real}}(\tx)=\hat{x}$, its gas price $v_{\mathrm{low}}$ is the minimum gas price that the transaction may be packaged. When this virtual transaction $\tx$ just satisfies the equation $\hat{p}^{\mathrm{real}}(\tx)=\hat{x}+1$, its gas price $v_{\mathrm{high}}$ is the minimum gas price that the transaction must be packaged by any block. We use the notation $\mathrm{sum}$ as the abbreviation of $\sum_{\tx'\in M}s(\tx')$, that the total block capacity usage of all transactions in mempool. By solving the two equations above, we obtain
\begin{align*}
v_{\mathrm{low}} & = \exp\left\{\frac{1}{\mathrm{sum}}\left(-k\lambda+\sum_{\tx'\in M}s(\tx')\ln v(\tx')\right)\right\}\\
v_{\mathrm{high}} & = \exp\left\{\frac{1}{\mathrm{sum}}\left((\mathrm{sum}-k)\lambda+\sum_{\tx'\in M}s(\tx')\ln v(\tx')\right)\right\}.
\end{align*}

The meanings of $v_{\mathrm{low}}$ and $v_{\mathrm{high}}$ are
\begin{itemize}
    \item $v_{\mathrm{low}}$: if an order's gas price is less than $v_{\mathrm{low}}$ then it cannot be included in a block;
    \item $v_{\mathrm{high}}$: if an order's gas price bidding is more than $v_{\mathrm{high}}$, then it will definitely be included in each block (within the network latency).
\end{itemize}
Therefore, $v_{\mathrm{low}}$ can be treated as the base fee like Ethereum. Moreover, $v_{\mathrm{high}}$ provides an upper bound on the gas price for the orders that are eager to be packaged. Specifically, a transaction with gas price $v_{\mathrm{high}}$ can force every potential miner to package it.

\section{Conclusion}\label{sec:conclu}

In this paper, we provide a framework to study the strategic behaviors of the transaction packaging game in high-throughput blockchains. We adopt a tractable two-period blockchain model to formalize high-throughput blockchains, which exhibit the incoordination and misallocation caused by network latency. Specifically, we investigate the effects of transaction costs, coordination, and market equilibrium on the decentralized platform and find that the strategic equilibrium we proposed can provide a more efficient and stable market outcome than the existing method of transaction packaging in high-throughput platforms.

Furthermore, we discuss the implications of our results for the development of decentralized multi-sided markets, and we suggest possible directions for future research. We argue that the decentralized platform has the potential to revolutionize the way markets are coordinated and allocated. We believe our work is an important step towards unlocking the full potential of decentralized multi-sided platforms.

Finally, the theoretical results of this paper also indicate several potential research directions. First, this paper builds up a benchmark model in a static manner, in which the mempool $M$ is static. Intuitively, a dynamic model should be considered in which the mempool $M$ changes over time. Another direction is to focus on the partial observability of mempool $M$ in reality. In this paper, we assume that every miner can observe the whole mempool $M$. However, in reality, each miner can only observe a subset of the mempool due to her network latency. Thus, it is interesting to study the optimal strategy of transaction packaging under partial observability of mempool. As a pioneer in the study of transaction packaging in high-throughput blockchains, we model network latency as consistent. In order to accurately model network latency in high-throughput blockchains, it is reasonable to consider the various factors that can affect the latency of orders and blocks.

\clearpage
\bibliographystyle{ACM-Reference-Format}
\bibliography{ref}


\begin{thebibliography}{35}


\ifx \showCODEN    \undefined \def \showCODEN     #1{\unskip}     \fi
\ifx \showDOI      \undefined \def \showDOI       #1{#1}\fi
\ifx \showISBNx    \undefined \def \showISBNx     #1{\unskip}     \fi
\ifx \showISBNxiii \undefined \def \showISBNxiii  #1{\unskip}     \fi
\ifx \showISSN     \undefined \def \showISSN      #1{\unskip}     \fi
\ifx \showLCCN     \undefined \def \showLCCN      #1{\unskip}     \fi
\ifx \shownote     \undefined \def \shownote      #1{#1}          \fi
\ifx \showarticletitle \undefined \def \showarticletitle #1{#1}   \fi
\ifx \showURL      \undefined \def \showURL       {\relax}        \fi
\providecommand\bibfield[2]{#2}
\providecommand\bibinfo[2]{#2}
\providecommand\natexlab[1]{#1}
\providecommand\showeprint[2][]{arXiv:#2}

\bibitem[Armstrong(2006)]%
        {armstrong2006competition}
\bibfield{author}{\bibinfo{person}{Mark Armstrong}.}
  \bibinfo{year}{2006}\natexlab{}.
\newblock \showarticletitle{Competition in two-sided markets}.
\newblock \bibinfo{journal}{\emph{The RAND journal of economics}}
  \bibinfo{volume}{37}, \bibinfo{number}{3} (\bibinfo{year}{2006}),
  \bibinfo{pages}{668--691}.
\newblock


\bibitem[Biais et~al\mbox{.}(2019)]%
        {biais2019blockchain}
\bibfield{author}{\bibinfo{person}{Bruno Biais}, \bibinfo{person}{Christophe
  Bisiere}, \bibinfo{person}{Matthieu Bouvard}, {and}
  \bibinfo{person}{Catherine Casamatta}.} \bibinfo{year}{2019}\natexlab{}.
\newblock \showarticletitle{The blockchain folk theorem}.
\newblock \bibinfo{journal}{\emph{The Review of Financial Studies}}
  \bibinfo{volume}{32}, \bibinfo{number}{5} (\bibinfo{year}{2019}),
  \bibinfo{pages}{1662--1715}.
\newblock


\bibitem[Buterin et~al\mbox{.}(2014)]%
        {buterin2014next}
\bibfield{author}{\bibinfo{person}{Vitalik Buterin} {et~al\mbox{.}}}
  \bibinfo{year}{2014}\natexlab{}.
\newblock \showarticletitle{A next-generation smart contract and decentralized
  application platform}.
\newblock \bibinfo{journal}{\emph{white paper}} \bibinfo{volume}{3},
  \bibinfo{number}{37} (\bibinfo{year}{2014}), \bibinfo{pages}{2--1}.
\newblock


\bibitem[Caillaud and Jullien(2003)]%
        {caillaud2003chicken}
\bibfield{author}{\bibinfo{person}{Bernard Caillaud} {and}
  \bibinfo{person}{Bruno Jullien}.} \bibinfo{year}{2003}\natexlab{}.
\newblock \showarticletitle{Chicken \& egg: Competition among intermediation
  service providers}.
\newblock \bibinfo{journal}{\emph{RAND journal of Economics}}
  (\bibinfo{year}{2003}), \bibinfo{pages}{309--328}.
\newblock


\bibitem[Chen and Micali(2019)]%
        {chen2019algorand}
\bibfield{author}{\bibinfo{person}{Jing Chen} {and} \bibinfo{person}{Silvio
  Micali}.} \bibinfo{year}{2019}\natexlab{}.
\newblock \showarticletitle{Algorand: A secure and efficient distributed
  ledger}.
\newblock \bibinfo{journal}{\emph{Theoretical Computer Science}}
  \bibinfo{volume}{777} (\bibinfo{year}{2019}), \bibinfo{pages}{155--183}.
\newblock


\bibitem[Chiu and Koeppl(2019)]%
        {chiu2019blockchain}
\bibfield{author}{\bibinfo{person}{Jonathan Chiu} {and}
  \bibinfo{person}{Thorsten~V Koeppl}.} \bibinfo{year}{2019}\natexlab{}.
\newblock \showarticletitle{Blockchain-based settlement for asset trading}.
\newblock \bibinfo{journal}{\emph{The Review of Financial Studies}}
  \bibinfo{volume}{32}, \bibinfo{number}{5} (\bibinfo{year}{2019}),
  \bibinfo{pages}{1716--1753}.
\newblock


\bibitem[Choi et~al\mbox{.}(2018)]%
        {choi2018fantom}
\bibfield{author}{\bibinfo{person}{Sang-Min Choi}, \bibinfo{person}{Jiho Park},
  \bibinfo{person}{Quan Nguyen}, {and} \bibinfo{person}{Andre Cronje}.}
  \bibinfo{year}{2018}\natexlab{}.
\newblock \showarticletitle{Fantom: A scalable framework for asynchronous
  distributed systems}.
\newblock \bibinfo{journal}{\emph{arXiv preprint arXiv:1810.10360}}
  (\bibinfo{year}{2018}).
\newblock


\bibitem[Cong and He(2019)]%
        {cong2019blockchain}
\bibfield{author}{\bibinfo{person}{Lin~William Cong} {and}
  \bibinfo{person}{Zhiguo He}.} \bibinfo{year}{2019}\natexlab{}.
\newblock \showarticletitle{Blockchain disruption and smart contracts}.
\newblock \bibinfo{journal}{\emph{The Review of Financial Studies}}
  \bibinfo{volume}{32}, \bibinfo{number}{5} (\bibinfo{year}{2019}),
  \bibinfo{pages}{1754--1797}.
\newblock


\bibitem[Cong et~al\mbox{.}(2021a)]%
        {cong2021decentralized}
\bibfield{author}{\bibinfo{person}{Lin~William Cong}, \bibinfo{person}{Zhiguo
  He}, {and} \bibinfo{person}{Jiasun Li}.} \bibinfo{year}{2021}\natexlab{a}.
\newblock \showarticletitle{Decentralized mining in centralized pools}.
\newblock \bibinfo{journal}{\emph{The Review of Financial Studies}}
  \bibinfo{volume}{34}, \bibinfo{number}{3} (\bibinfo{year}{2021}),
  \bibinfo{pages}{1191--1235}.
\newblock


\bibitem[Cong et~al\mbox{.}(2021b)]%
        {cong2021tokenomics}
\bibfield{author}{\bibinfo{person}{Lin~William Cong}, \bibinfo{person}{Ye Li},
  {and} \bibinfo{person}{Neng Wang}.} \bibinfo{year}{2021}\natexlab{b}.
\newblock \showarticletitle{Tokenomics: Dynamic adoption and valuation}.
\newblock \bibinfo{journal}{\emph{The Review of Financial Studies}}
  \bibinfo{volume}{34}, \bibinfo{number}{3} (\bibinfo{year}{2021}),
  \bibinfo{pages}{1105--1155}.
\newblock


\bibitem[Cong et~al\mbox{.}(2022)]%
        {cong2022token}
\bibfield{author}{\bibinfo{person}{Lin~William Cong}, \bibinfo{person}{Ye Li},
  {and} \bibinfo{person}{Neng Wang}.} \bibinfo{year}{2022}\natexlab{}.
\newblock \showarticletitle{Token-based platform finance}.
\newblock \bibinfo{journal}{\emph{Journal of Financial Economics}}
  \bibinfo{volume}{144}, \bibinfo{number}{3} (\bibinfo{year}{2022}),
  \bibinfo{pages}{972--991}.
\newblock


\bibitem[Daian et~al\mbox{.}(2020)]%
        {daian2020flash}
\bibfield{author}{\bibinfo{person}{Philip Daian}, \bibinfo{person}{Steven
  Goldfeder}, \bibinfo{person}{Tyler Kell}, \bibinfo{person}{Yunqi Li},
  \bibinfo{person}{Xueyuan Zhao}, \bibinfo{person}{Iddo Bentov},
  \bibinfo{person}{Lorenz Breidenbach}, {and} \bibinfo{person}{Ari Juels}.}
  \bibinfo{year}{2020}\natexlab{}.
\newblock \showarticletitle{Flash boys 2.0: Frontrunning in decentralized
  exchanges, miner extractable value, and consensus instability}. In
  \bibinfo{booktitle}{\emph{2020 IEEE Symposium on Security and Privacy (SP)}}.
  IEEE, \bibinfo{pages}{910--927}.
\newblock


\bibitem[Easley et~al\mbox{.}(2019)]%
        {easley2019mining}
\bibfield{author}{\bibinfo{person}{David Easley}, \bibinfo{person}{Maureen
  O'Hara}, {and} \bibinfo{person}{Soumya Basu}.}
  \bibinfo{year}{2019}\natexlab{}.
\newblock \showarticletitle{From mining to markets: The evolution of bitcoin
  transaction fees}.
\newblock \bibinfo{journal}{\emph{Journal of Financial Economics}}
  \bibinfo{volume}{134}, \bibinfo{number}{1} (\bibinfo{year}{2019}),
  \bibinfo{pages}{91--109}.
\newblock


\bibitem[Gilad et~al\mbox{.}(2017)]%
        {gilad2017algorand}
\bibfield{author}{\bibinfo{person}{Yossi Gilad}, \bibinfo{person}{Rotem Hemo},
  \bibinfo{person}{Silvio Micali}, \bibinfo{person}{Georgios Vlachos}, {and}
  \bibinfo{person}{Nickolai Zeldovich}.} \bibinfo{year}{2017}\natexlab{}.
\newblock \showarticletitle{Algorand: Scaling byzantine agreements for
  cryptocurrencies}. In \bibinfo{booktitle}{\emph{Proceedings of the 26th
  symposium on operating systems principles}}. \bibinfo{pages}{51--68}.
\newblock


\bibitem[Gryglewicz et~al\mbox{.}(2021)]%
        {gryglewicz2021optimal}
\bibfield{author}{\bibinfo{person}{Sebastian Gryglewicz},
  \bibinfo{person}{Simon Mayer}, {and} \bibinfo{person}{Erwan Morellec}.}
  \bibinfo{year}{2021}\natexlab{}.
\newblock \showarticletitle{Optimal financing with tokens}.
\newblock \bibinfo{journal}{\emph{Journal of Financial Economics}}
  \bibinfo{volume}{142}, \bibinfo{number}{3} (\bibinfo{year}{2021}),
  \bibinfo{pages}{1038--1067}.
\newblock


\bibitem[Jullien and Pavan(2019)]%
        {jullien2019information}
\bibfield{author}{\bibinfo{person}{Bruno Jullien} {and}
  \bibinfo{person}{Alessandro Pavan}.} \bibinfo{year}{2019}\natexlab{}.
\newblock \showarticletitle{Information management and pricing in platform
  markets}.
\newblock \bibinfo{journal}{\emph{The Review of Economic Studies}}
  \bibinfo{volume}{86}, \bibinfo{number}{4} (\bibinfo{year}{2019}),
  \bibinfo{pages}{1666--1703}.
\newblock


\bibitem[Kiayias and Russell(2018)]%
        {kiayias2018ouroboros}
\bibfield{author}{\bibinfo{person}{Aggelos Kiayias} {and}
  \bibinfo{person}{Alexander Russell}.} \bibinfo{year}{2018}\natexlab{}.
\newblock \showarticletitle{Ouroboros-bft: A simple byzantine fault tolerant
  consensus protocol}.
\newblock \bibinfo{journal}{\emph{Cryptology ePrint Archive}}
  (\bibinfo{year}{2018}).
\newblock


\bibitem[Kiayias et~al\mbox{.}(2017)]%
        {kiayias2017ouroboros}
\bibfield{author}{\bibinfo{person}{Aggelos Kiayias}, \bibinfo{person}{Alexander
  Russell}, \bibinfo{person}{Bernardo David}, {and} \bibinfo{person}{Roman
  Oliynykov}.} \bibinfo{year}{2017}\natexlab{}.
\newblock \showarticletitle{Ouroboros: A provably secure proof-of-stake
  blockchain protocol}. In \bibinfo{booktitle}{\emph{Annual international
  cryptology conference}}. Springer, \bibinfo{pages}{357--388}.
\newblock


\bibitem[Leonardos et~al\mbox{.}(2021)]%
        {leonardos2021dynamical}
\bibfield{author}{\bibinfo{person}{Stefanos Leonardos},
  \bibinfo{person}{Barnab{\'e} Monnot}, \bibinfo{person}{Dani{\"e}l
  Reijsbergen}, \bibinfo{person}{Efstratios Skoulakis}, {and}
  \bibinfo{person}{Georgios Piliouras}.} \bibinfo{year}{2021}\natexlab{}.
\newblock \showarticletitle{Dynamical analysis of the eip-1559 ethereum fee
  market}. In \bibinfo{booktitle}{\emph{Proceedings of the 3rd ACM Conference
  on Advances in Financial Technologies}}. \bibinfo{pages}{114--126}.
\newblock


\bibitem[Li et~al\mbox{.}(2020)]%
        {li2020decentralized}
\bibfield{author}{\bibinfo{person}{Chenxin Li}, \bibinfo{person}{Peilun Li},
  \bibinfo{person}{Dong Zhou}, \bibinfo{person}{Zhe Yang},
  \bibinfo{person}{Ming Wu}, \bibinfo{person}{Guang Yang}, \bibinfo{person}{Wei
  Xu}, \bibinfo{person}{Fan Long}, {and} \bibinfo{person}{Andrew Chi-Chih
  Yao}.} \bibinfo{year}{2020}\natexlab{}.
\newblock \showarticletitle{A decentralized blockchain with high throughput and
  fast confirmation}. In \bibinfo{booktitle}{\emph{2020 $\{$USENIX$\}$ Annual
  Technical Conference ($\{$USENIX$\}$$\{$ATC$\}$ 20)}}.
  \bibinfo{pages}{515--528}.
\newblock


\bibitem[Liu et~al\mbox{.}(2022)]%
        {10.1145/3548606.3559341}
\bibfield{author}{\bibinfo{person}{Yulin Liu}, \bibinfo{person}{Yuxuan Lu},
  \bibinfo{person}{Kartik Nayak}, \bibinfo{person}{Fan Zhang},
  \bibinfo{person}{Luyao Zhang}, {and} \bibinfo{person}{Yinhong Zhao}.}
  \bibinfo{year}{2022}\natexlab{}.
\newblock \showarticletitle{Empirical Analysis of EIP-1559: Transaction Fees,
  Waiting Times, and Consensus Security}. In
  \bibinfo{booktitle}{\emph{Proceedings of the 2022 ACM SIGSAC Conference on
  Computer and Communications Security}} (Los Angeles, CA, USA)
  \emph{(\bibinfo{series}{CCS '22})}. \bibinfo{publisher}{Association for
  Computing Machinery}, \bibinfo{address}{New York, NY, USA},
  \bibinfo{pages}{2099–2113}.
\newblock
\showISBNx{9781450394505}
\urldef\tempurl%
\url{https://doi.org/10.1145/3548606.3559341}
\showDOI{\tempurl}


\bibitem[Nakamoto(2008)]%
        {nakamoto2008bitcoin}
\bibfield{author}{\bibinfo{person}{Satoshi Nakamoto}.}
  \bibinfo{year}{2008}\natexlab{}.
\newblock \showarticletitle{Bitcoin: A peer-to-peer electronic cash system}.
\newblock \bibinfo{journal}{\emph{Decentralized Business Review}}
  (\bibinfo{year}{2008}), \bibinfo{pages}{21260}.
\newblock


\bibitem[Prat and Walter(2021)]%
        {prat2021equilibrium}
\bibfield{author}{\bibinfo{person}{Julien Prat} {and} \bibinfo{person}{Benjamin
  Walter}.} \bibinfo{year}{2021}\natexlab{}.
\newblock \showarticletitle{An equilibrium model of the market for bitcoin
  mining}.
\newblock \bibinfo{journal}{\emph{Journal of Political Economy}}
  \bibinfo{volume}{129}, \bibinfo{number}{8} (\bibinfo{year}{2021}),
  \bibinfo{pages}{2415--2452}.
\newblock


\bibitem[Qin et~al\mbox{.}(2022)]%
        {qin2022quantifying}
\bibfield{author}{\bibinfo{person}{Kaihua Qin}, \bibinfo{person}{Liyi Zhou},
  {and} \bibinfo{person}{Arthur Gervais}.} \bibinfo{year}{2022}\natexlab{}.
\newblock \showarticletitle{Quantifying blockchain extractable value: How dark
  is the forest?}. In \bibinfo{booktitle}{\emph{2022 IEEE Symposium on Security
  and Privacy (SP)}}. IEEE, \bibinfo{pages}{198--214}.
\newblock


\bibitem[Rochet and Tirole(2003)]%
        {rochet2003platform}
\bibfield{author}{\bibinfo{person}{Jean-Charles Rochet} {and}
  \bibinfo{person}{Jean Tirole}.} \bibinfo{year}{2003}\natexlab{}.
\newblock \showarticletitle{Platform competition in two-sided markets}.
\newblock \bibinfo{journal}{\emph{Journal of the european economic
  association}} \bibinfo{volume}{1}, \bibinfo{number}{4}
  (\bibinfo{year}{2003}), \bibinfo{pages}{990--1029}.
\newblock


\bibitem[Rochet and Tirole(2006)]%
        {rochet2006two}
\bibfield{author}{\bibinfo{person}{Jean-Charles Rochet} {and}
  \bibinfo{person}{Jean Tirole}.} \bibinfo{year}{2006}\natexlab{}.
\newblock \showarticletitle{Two-sided markets: a progress report}.
\newblock \bibinfo{journal}{\emph{The RAND journal of economics}}
  \bibinfo{volume}{37}, \bibinfo{number}{3} (\bibinfo{year}{2006}),
  \bibinfo{pages}{645--667}.
\newblock


\bibitem[Roughgarden(2021)]%
        {roughgarden2021transaction}
\bibfield{author}{\bibinfo{person}{Tim Roughgarden}.}
  \bibinfo{year}{2021}\natexlab{}.
\newblock \showarticletitle{Transaction fee mechanism design}.
\newblock \bibinfo{journal}{\emph{ACM SIGecom Exchanges}} \bibinfo{volume}{19},
  \bibinfo{number}{1} (\bibinfo{year}{2021}), \bibinfo{pages}{52--55}.
\newblock


\bibitem[Saleh(2021)]%
        {saleh2021blockchain}
\bibfield{author}{\bibinfo{person}{Fahad Saleh}.}
  \bibinfo{year}{2021}\natexlab{}.
\newblock \showarticletitle{Blockchain without waste: Proof-of-stake}.
\newblock \bibinfo{journal}{\emph{The Review of financial studies}}
  \bibinfo{volume}{34}, \bibinfo{number}{3} (\bibinfo{year}{2021}),
  \bibinfo{pages}{1156--1190}.
\newblock


\bibitem[Schilling and Uhlig(2019)]%
        {schilling2019some}
\bibfield{author}{\bibinfo{person}{Linda Schilling} {and}
  \bibinfo{person}{Harald Uhlig}.} \bibinfo{year}{2019}\natexlab{}.
\newblock \showarticletitle{Some simple bitcoin economics}.
\newblock \bibinfo{journal}{\emph{Journal of Monetary Economics}}
  \bibinfo{volume}{106} (\bibinfo{year}{2019}), \bibinfo{pages}{16--26}.
\newblock


\bibitem[Sockin and Xiong(2022)]%
        {sockin2022decentralization}
\bibfield{author}{\bibinfo{person}{Michael Sockin} {and} \bibinfo{person}{Wei
  Xiong}.} \bibinfo{year}{2022}\natexlab{}.
\newblock \bibinfo{booktitle}{\emph{Decentralization through tokenization}}.
\newblock \bibinfo{type}{{T}echnical {R}eport}. \bibinfo{pages}{forthcoming}
  pages.
\newblock


\bibitem[Spiegelman et~al\mbox{.}(2022)]%
        {spiegelman2022bullshark}
\bibfield{author}{\bibinfo{person}{Alexander Spiegelman}, \bibinfo{person}{Neil
  Giridharan}, \bibinfo{person}{Alberto Sonnino}, {and}
  \bibinfo{person}{Lefteris Kokoris-Kogias}.} \bibinfo{year}{2022}\natexlab{}.
\newblock \showarticletitle{Bullshark: Dag bft protocols made practical}. In
  \bibinfo{booktitle}{\emph{Proceedings of the 2022 ACM SIGSAC Conference on
  Computer and Communications Security}}. \bibinfo{pages}{2705--2718}.
\newblock


\bibitem[Tan and Zhou(2021)]%
        {tan2021effects}
\bibfield{author}{\bibinfo{person}{Guofu Tan} {and} \bibinfo{person}{Junjie
  Zhou}.} \bibinfo{year}{2021}\natexlab{}.
\newblock \showarticletitle{The effects of competition and entry in multi-sided
  markets}.
\newblock \bibinfo{journal}{\emph{The Review of Economic Studies}}
  \bibinfo{volume}{88}, \bibinfo{number}{2} (\bibinfo{year}{2021}),
  \bibinfo{pages}{1002--1030}.
\newblock


\bibitem[Teh(2022)]%
        {teh2022platform}
\bibfield{author}{\bibinfo{person}{Tat-How Teh}.}
  \bibinfo{year}{2022}\natexlab{}.
\newblock \showarticletitle{Platform governance}.
\newblock \bibinfo{journal}{\emph{American Economic Journal: Microeconomics}}
  \bibinfo{volume}{14}, \bibinfo{number}{3} (\bibinfo{year}{2022}),
  \bibinfo{pages}{213--54}.
\newblock


\bibitem[Weyl(2010)]%
        {weyl2010price}
\bibfield{author}{\bibinfo{person}{E~Glen Weyl}.}
  \bibinfo{year}{2010}\natexlab{}.
\newblock \showarticletitle{A price theory of multi-sided platforms}.
\newblock \bibinfo{journal}{\emph{American Economic Review}}
  \bibinfo{volume}{100}, \bibinfo{number}{4} (\bibinfo{year}{2010}),
  \bibinfo{pages}{1642--72}.
\newblock


\bibitem[Yakovenko(2018)]%
        {yakovenko2018solana}
\bibfield{author}{\bibinfo{person}{Anatoly Yakovenko}.}
  \bibinfo{year}{2018}\natexlab{}.
\newblock \showarticletitle{Solana: A new architecture for a high performance
  blockchain v0. 8.13}.
\newblock \bibinfo{journal}{\emph{Whitepaper}} (\bibinfo{year}{2018}).
\newblock


\end{thebibliography}
\clearpage

\appendix
\section{Proof of Theorem 4.1}\label{sec:proofreal}

\thmainreal*

\begin{proof}

Recall that $k$ is the total capacity of the block, $v(\tx)$ is the gas price of the transaction $\tx$, and  $s(\tx)$ is the block capacity used by transaction $\tx$. It can be verified that the proof of \Cref{lemma:eqcond} is still valid under the setting. We use the notation $\mathrm{sum}$ as the abbreviation of $\sum_{\tx'\in M}s(\tx')$, that the total block capacity usage of all transactions in mempool. 

Like \Cref{lemma:weak}, we have such two facts about the function $\hat{p}^{\mathrm{real}}(\cdot)$:

\begin{small}\begin{align*}
& \sum_{\tx\in M} s(\tx)\hat{p}^{\mathrm{real}}(\tx)\\
= & \sum_{\tx\in M} \frac{s(\tx)}{\mathrm{sum}}\left(k+\sum_{\tx'\in M} s(\tx')\frac{\ln v(\tx)-\ln v(\tx')}{\lambda}\right)\\
= & \left(\sum_{\tx\in M}\frac{s(\tx)k}{\mathrm{sum}}\right)+\frac{s(\tx)}{\mathrm{sum}}\sum_{\tx\in M} \sum_{\tx'\in M} s(\tx')\frac{\ln v(\tx)-\ln v(\tx')}{\lambda}\\
= & k + \frac{s(\tx)}{\mathrm{sum}}\sum_{\tx\in M} \left(\mathrm{sum}\frac{\ln v(\tx)}{\lambda}\sum_{\tx'\in M} \frac{-(\ln v(\tx'))s(\tx')}{\lambda}\right)\\
= & k + \frac{s(\tx)}{\mathrm{sum}}\left(\left(\sum_{\tx\in M} \frac{\mathrm{sum}\ln v(\tx)}{\lambda}\right)-\left(\mathrm{sum}\sum_{\tx'\in M} \frac{-\ln v(\tx')}{\lambda}\right)\right)\\
= & k,
\end{align*}\end{small}

\begin{small}\begin{align*}
& v(\tx)\exp\left(-\lambda \hat{p}^{\mathrm{real}}(\tx)\right)\\
= & v(\tx)\exp\left(- \frac{\lambda}{\mathrm{sum}}\left(k+\sum_{\tx'\in M} s(\tx')\frac{\ln v(\tx)-\ln v(\tx')}{\lambda}\right)\right)\\
= & v(\tx)\exp\left(-\frac{\lambda k}{\mathrm{sum}}\right)\exp\left(-\frac{\lambda}{\mathrm{sum}}\sum_{\tx'\in M} s(\tx')\frac{\ln v(\tx)-\ln v(\tx')}{\lambda}\right)\\
= & v(\tx)\exp\left(-\frac{\lambda k}{\mathrm{sum}}\right)\exp\left(-\frac{1}{\mathrm{sum}}\left(\mathrm{sum}\ln v(\tx)-\sum_{\tx'\in M}s(\tx')\ln v(\tx')\right)\right)\\
= & v(\tx)\exp\left(-\frac{\lambda k}{\mathrm{sum}}\right)\exp\left(-\ln v(\tx)+\frac{1}{\mathrm{sum}}\sum_{\tx'\in M}s(\tx')\ln v(\tx')\right)\\
= & v(\tx)\exp\left(-\frac{\lambda k}{\mathrm{sum}}\right)\frac{1}{v(\tx)}\exp\left(\frac{1}{\mathrm{sum}}\sum_{\tx'\in M}s(\tx')\ln v(\tx')\right)\\
= & \exp\left(-\frac{\lambda k}{\mathrm{sum}}\right)\exp\left(\frac{1}{\mathrm{sum}}\sum_{\tx'\in M}s(\tx')\ln v(\tx')\right).
\end{align*}\end{small}

Then, like \Cref{th:algeq}, the rest of the proof has three parts:
\begin{enumerate}
    \item to show the equation $\sum_{\tx\in M}\min(\max(\hat{p}^{\mathrm{real}}(\tx)-x,0),1)=k$ always has a solution;
    \item to show $p_{\hat{x}}^{\mathrm{real}}(\cdot)$ satisfies that (1) $\forall\tx\in M, p_{\hat{x}}^{\mathrm{real}}(\tx)\in [0,1]$; and (2) $\sum_{\tx\in M}p_{\hat{x}}^{\mathrm{real}}(\tx)s(\tx)=k$;
    \item to prove that the corresponding strategy $\sigma$ of function $p_{\hat{x}}^{\mathrm{real}}(\cdot)$ satisfies the condition of \Cref{lemma:eqcond}.
\end{enumerate}

Part (1) and (2) are exactly the same as the proof of \Cref{th:algeq}, so we omit them. In part (3), we set
\[
w = \exp\left(-\frac{\lambda k}{\mathrm{sum}}\right)\exp\left(\frac{1}{\mathrm{sum}}\sum_{\tx'\in M}s(\tx')\ln v(\tx')\right)\exp(-\lambda \hat{x})
\] and obtain that for all $\tx\in M$,
\begin{align*}
& v(\tx)\exp\left(-\lambda \left(\hat{p}^{\mathrm{real}}(\tx)-x\right)\right)\\
= & v(\tx)\exp\left(-\lambda \hat{p}^{\mathrm{real}}(\tx)\right)\exp(-\lambda x)\\
= & \exp\left(-\frac{\lambda k}{\mathrm{sum}}\right)\exp\left(\frac{1}{\mathrm{sum}}\sum_{\tx'\in M}s(\tx')\ln v(\tx')\right)\\
= & w.
\end{align*}
Thus,
\[
\begin{cases}
p_{\hat{x}}^{\mathrm{real}}(\tx)=0 \implies v(\tx)\exp\left(-\lambda p_{\hat{x}}^{\mathrm{real}}(\tx)\right)<v(\tx)\exp\left(-\lambda \left(\hat{p}^{\mathrm{real}}(\tx)-\hat{x}\right)\right)=w\\
0<p_{\hat{x}}^{\mathrm{real}}(\tx)<1 \implies v(\tx)\exp\left(-\lambda p_{\hat{x}}^{\mathrm{real}}(\tx)\right)=v(\tx)\exp\left(-\lambda \left(\hat{p}^{\mathrm{real}}(\tx)-\hat{x}\right)\right)=w\\
p_{\hat{x}}^{\mathrm{real}}(\tx)=1 \implies v(\tx)\exp\left(-\lambda p_{\hat{x}}^{\mathrm{real}}(\tx)\right)>v(\tx)\exp\left(-\lambda \left(\hat{p}^{\mathrm{real}}(\tx)-\hat{x}\right)\right)=w,
\end{cases}
\] thereby meeting the condition of \Cref{lemma:eqcond}. So, if there is an mixed strategy $\sigma$ with $p^{\sigma}(\tx)=p_{\hat{x}}^{\mathrm{real}}(\tx)$ for each $\tx\in M$, then $\sigma$ is an equilibrium strategy.
\end{proof}

\section{Algorithm for Calculating $\hat{x}$}\label{sec:algcalcx}
The goal of the algorithm is to calculate the smallest solution of equation
\[
\sum_{j=1}^m \min(\max(p[j]-x,0),1)\times s[j]=k,
\]
where $p[1..m]$ and $s[1..m]$ are specific non-negative arrays.

\begin{algorithm}[htbp]
    \SetAlgoNoLine
    \KwIn{array $p[1..m]$, array $s[1..m]$, parameter $k$.}
    \KwResult{$\hat{x}$, the smallest solution of equation $\sum_{j=1}^m \min(\max(p[j]-x,0),1)s[j]=k$}
    New array $b[1..2m]$\;
    \For{$j\in\set{1,2,...,m}$}
    {
        $b[j]\gets p[j]$\;
        $b[j+m]\gets p[j]-1$\;
    }
    Sort the $b[]$ array from small to large\;
    Define function $f:\mathbb{R}\rightarrow\mathbb{R}$ as $f(x)=\sum_{j=1}^m \min(\max(p[j]-x,0),1)\times s[j]$\;
    \tcc{$f(\cdot)$ is continuous, non-increasing, and piece-wise linear.}
    \tcc{$f(\cdot)$ is linear in each $[b[j],b[j+1]]$ interval.}
    $l\gets 1$; $r\gets 2m$\;
    \While{$l+1\neq r$}
    {
        $mid\gets \lfloor \frac{l+r}{2}\rfloor$\;
        \uIf{$f(b[mid])>k$}
        {
            $l\gets mid$\;
        }
        \Else
        {
            $r\gets mid$\;
        }
    }
    \tcc{We obtain $f(b[l])> k$ and $f(b[l+1])\leq k$.}
    $pos=\frac{b[l]+b[l+1]}{2}$; $d\gets 0$\;
    \For{$j\in\set{1,2,...,m}$}
    {
        \If{$pos<p[j]<pos+1$}
        {
            $d\gets d-s[j]$
        }
    }
    \tcc{$d$ is the slope of the function $f$ at $(b[l],b[l+1])$.}
    \KwRet{$b[l]+\frac{k-b[l]}{d}$}\;
    \caption{Calculating $\hat{x}$}
    \label{alg:calcx}
\end{algorithm}

\end{document}